\documentclass{article}
\usepackage{graphicx}
\usepackage{amsfonts}
\usepackage{amsmath}
\usepackage{amssymb}
\usepackage{url}
\usepackage{tikz}
\usepackage{pgfplots}
\pgfplotsset{compat=1.17}
\usepackage{dsfont} 

\usepackage{indentfirst}
\usepackage{enumerate}
\usepackage[normalem]{ulem}
\usepackage{mathrsfs}
\usepackage{multirow}
\usepackage{diagbox}
\usepackage{bbold}
\usepackage{stackengine}

 \usepackage[colorlinks=true,citecolor=blue]{hyperref}
\usepackage{amsthm}
\usepackage{color}

\definecolor{DarkGreen}{rgb}{0.2,0.6,0.2}

\definecolor{purple}{rgb}{0.6,0.3,0.8}

\usepackage{natbib} 
\usepackage{comment}
\allowdisplaybreaks

\def\d{\mathrm{d}}

\newcommand{\E}{\mathbb{E}}

\newcommand{\R}{\mathbb{R}}

\newcommand{\dsquare}{\mathop{  \square} \displaylimits}

\newcommand{\p}{\mathbb{P}}
\newcommand{\X}{\mathcal{X}}

\renewcommand{\H}{\mathcal{H}}

\newcommand{\id}{\mathds{1}}

\renewcommand{\ge}{\geqslant}
\renewcommand{\le}{\leqslant}
\renewcommand{\geq}{\geqslant}
\renewcommand{\leq}{\leqslant}
\renewcommand{\epsilon}{\varepsilon}
\newcommand{\esssup}{\mathrm{ess\mbox{-}sup}}
\newcommand{\essinf}{\mathrm{ess\mbox{-}inf}}

\theoremstyle{plain}
\newtheorem{theorem}{Theorem}
\newtheorem{corollary}{Corollary}
\newtheorem{lemma}{Lemma}
\newtheorem{proposition}{Proposition}
\theoremstyle{definition}
\newtheorem{definition}{Definition}
\newtheorem{example}{Example}
\newtheorem{assumption}{Assumption}

\theoremstyle{remark}
\newtheorem{remark}{Remark}

\newcommand{\VaR}{\mathrm{VaR}}

\newcommand{\ES}{\mathrm{ES}}

\newcommand{\dboxminus}{\mathop{  \boxminus} \displaylimits}
\newcommand{\dboxplus}{\mathop{  \boxplus} \displaylimits}

\usepackage{setspace}

\usepackage{footmisc}
\title{Counter-monotonic risk allocations and distortion risk measures} 
\author{
Mario Ghossoub\thanks%
  {Department of Statistics and Actuarial Science,
  University of Waterloo,
  Waterloo, Ontario, Canada.
  E-mail: \href{mailto:mario.ghossoub@uwaterloo.ca}{
mario.ghossoub@uwaterloo.ca}.}
  \and Qinghua Ren\thanks%
  {Department of Statistics and Actuarial Science,
  University of Waterloo,
  Waterloo, Ontario, Canada.
  E-mail: \href{mailto:qinghua.ren@uwaterloo.ca}{qinghua.ren@uwaterloo.ca}.}
  \and Ruodu Wang\thanks%
  {Department of Statistics and Actuarial Science,
  University of Waterloo,
  Waterloo, Ontario, Canada.
  E-mail: \href{mailto:wang@uwaterloo.ca}{wang@uwaterloo.ca}.}
  }

\begin{document}
\maketitle

\begin{abstract} 
In risk-sharing markets with aggregate uncertainty, characterizing Pareto-optimal allocations when agents might not be risk averse is a challenging task, and the literature has only provided limited explicit results thus far. In particular, Pareto optima in such a setting may not necessarily be comonotonic, in contrast to the case of risk-averse agents. In fact, when market participants are risk-seeking, Pareto-optimal allocations are counter-monotonic. Counter-monotonicity of Pareto optima also arises in some situations for quantile-optimizing agents. In this paper, we provide a systematic study of efficient risk sharing in markets where allocations are constrained to be counter-monotonic.  
The preferences of the agents are   modelled by a common distortion risk measure, or equivalently, by a common Yaari dual utility. We consider three different settings:  risk-averse agents, risk-seeking agents, and those with an inverse S-shaped distortion function. In each case, we provide useful characterizations of optimal allocations, for both the counter-monotonic market and the unconstrained market. 
%
%
To illustrate our results, we consider an application to a portfolio choice problem for a portfolio manager tasked with managing the investments of a group of clients, with varying levels of risk aversion or risk seeking. We determine explicitly the optimal investment strategies in this case. Our results confirm the intuition that a manager investing on behalf of risk-seeking agents tends to invest more in risky assets than a manager acting on behalf of risk-averse agents.
\medskip
\end{abstract}

\section{Introduction}

The literature on efficient allocations in pure exchange economies, or optimal risk sharing in risk-sharing markets, has hitherto mostly been interested in Pareto optimality (also called Pareto efficiency)  for risk-averse agents, both within the classical expected-utility theory (EUT) framework and beyond. For decision-making under objective risk, risk aversion is commonly defined as consistency with respect to second-order stochastic dominance, as argued by \cite{RothschildStiglitz1970}. The seminal work of  \cite{borch1962} and \cite{wilson1968theory} on risk sharing for risk-averse agents within EUT provided an explicit characterization of Pareto-optimal allocations, and showed in particular that optimal allocations are \emph{comonotonic}. This property of optima was extended beyond the EUT framework, and a cornerstone result in this literature is the so-called \textit{comonotonic improvement theorem} (e.g., \cite{landsberger1994co}, \cite{carlier2003core}, \cite{ruschendorf2013mathematical}, or \cite{denuitetal2023comonotonicity}). Existence and comonotonicity of Pareto optima beyond EUT has been established widely. See, for instance, \cite{chateauneuf2000optimal}, \cite{dana2004ambiguity}, \cite{tsanakas2006risk}, \cite{decastro2011ambiguity}, or \cite{beissner2023optimal}, for risk-sharing markets with ambiguity-sensitive agents; or \cite{barrieu2005inf}, \cite{jouini2008optimal}, and \cite{filipovic2008optimal} for the popular class of law-invariant monetary utilities, or law-invariant convex risk measures. \cite{RavanelliSvindland2014} showed the existence and comonotonicity of Pareto optima for a class of law-invariant variational preferences that are SSD-consistent. Recently, \cite{ghossoubzhu2024} provided an explicit characterization of Pareto optima and showed their comonotonicity, for a wide class of robust SSD-consistent concave utility functionals.

 Pareto optimality in risk-sharing markets that include risk-seeking agents is less well-studied. A first step in this direction was taken by \cite{ACGNecta2018, AGS2022} and \cite{HeringsZhany2022} in an equilibrium context. \cite{beissner2023optimal} allowed for the existence of risk-loving and ambiguity-loving agents in exchange economies. Recently, \cite{lauzier2023pairwise}   provided a stochastic representation of counter-monotonicity. Using this result, \cite{lauzier2024negatively} further derived the so-called \textit{counter-monotonic improvement theorem},   a counterpart to the comonotonic improvement theorem. Just as risk aversion is linked to comonotonicity, risk seeking is linked to counter-monotonicity. The counter-monotonic improvement theorem states that for any random vector bounded from below or above, there exists a counter-monotonic random vector whose components are riskier than those of the given random vector. An important implication of this theorem is that counter-monotonic allocations will always be preferred by risk-seeking agents. In addition to risk-seeking agents,  counter-monotonic allocations can also be optimal for quantile agents, who are neither risk-seeking nor risk-averse, as shown  by \cite{embrechts2018quantile, embrechts2020quantile} and generalized by \cite{weber2018solvency}. These observations suggest that a systemic study of risk sharing problems when constrained to counter-monotonic allocations is useful, and this is central to this paper.

Specifically, in this paper, we study optimal risk sharing and inf-convolution of distortion risk measures. We consider not only risk-averse agents, but also risk-seeking agents and behavioral agents. Our study can be  further motivated by a simple  
portfolio optimization problem. A fund collects an initial  constant endowment 
$W$ from a group of clients who have the same risk preferences. These clients can either be risk-seeking or risk-averse, and their risk attitudes are modeled using distortion risk measures $\rho_1,\dots,\rho_n$. The fund is managed by a professional manager whose task is to invest the endowment in a manner that aligns with the clients' risk attitudes and minimizes the total risk.
The manager intends to allocate a proportion $\lambda$
of the fund to a risky asset $X$ in a domain $\mathcal{X}$.
 In this context, Pareto optimality of an allocation $(X_1,\dots,X_n)$ is equivalent to  optimality with respect to the sum of the risk measures; see \citet[Proposition 1]{embrechts2018quantile}. 
 Therefore, the objective function of the fund manager is  given by 
 \begin{align}
     \label{eq:intro}
 \text{minimize}~~\sum_{i=1}^n \rho_{i}(-X_i) \mbox{~~subject to~}
 X_1+\dots + X_n=
 W+\lambda X-c(\lambda) \mbox{~and $c(\lambda)\le W$},\end{align}
where $c$ represents the corresponding cost function of investing $\lambda$ in the fund,
and hence  $W+\lambda X-c(\lambda)$ represents the terminal wealth of the fund, to be allocated to the participants. 
Note that there are two layers of  optimization involved in \eqref{eq:intro},
that is, deciding the investment strategy $\lambda $ and the allocation $(X_1,\dots,X_n)$.  
Investors interested in risky assets are typically not necessarily risk-averse, and hence \eqref{eq:intro} calls for a study of inf-convolution for non-risk-averse agents. 
As we will see in Section \ref{sec:port}, our results on risk sharing for distortion risk measures  will help to solve \eqref{eq:intro}. 
A natural intuition is that less will be invested in the  risky asset when risk-averse agents are involved, whereas more will be invested 
in the risky asset  when risk-seeking agents are involved.
This intuition checks out as justified by our results.



As mentioned above, we will focus on distortion risk measures, a popular class of risk measures widely used in finance and insurance; see \cite{mcneil2015quantitative} for distortion risk measures in risk management.
To make our notion more specialized, now we use $\rho_{h_1}, \dots, \rho_{h_n}$ for distortion risk measures and $h_1,\dots,h_n$ for their distortion functions.
We denote by  $\dsquare_{i=1}^{n}\rho_{h_{i}}(X)$,  $\dboxplus_{i=1}^{n}\rho_{h_{i}}(X)$
and $\dboxminus_{i=1}^{n}\rho_{h_{i}}(X)$
the smallest value of $\rho_{h_1} (X_1)+\dots+
\rho_{h_n}(X_n)$ over allocations of $X$ that are general, comonotonic, and counter-monotonic, respectively;  see Sections \ref{sec:pre} and \ref{sec:counter} for precise definitions.
 
The  existing literature primarily focuses on the unconstrained inf-convolution $\dsquare_{i=1}^{n}\rho_{h_{i}}$
and comontonic inf-convolution 
$\dboxplus_{i=1}^{n}\rho_{h_{i}}$.
The problem setting of counter-monotonic inf-convolution $\dboxminus_{i=1}^{n}\rho_{h_{i}}$ is novel.    For most results in our paper, we assume $h_1=\dots=h_n$, that is, the homogeneous setting, for tractability. 
We aim to answer the following questions:
\begin{enumerate}[(a)]
    \item  It is known that the comonotonic inf-convolution has an explicit form, i.e., $\dboxplus_{i=1}^{n}\rho_{h} = \rho_h$, which holds true regardless of whether $h$ is concave or convex; see  e.g., \citet[Proposition 5]{embrechts2018quantile}. 
   This naturally leads to the question of whether the counter-monotonic inf-convolution exhibits similar properties. Specifically, does $\dboxminus_{i=1}^{n}\rho_{h}$ share the same formula for both concave and convex $h$? Under which conditions does $\dboxminus_{i=1}^{n}\rho_{h}=\rho_h$ hold true?
    \item  We study three types of problems, i.e., unconstrained, comonotonic, and counter-monotonic risk sharing. Each type provides  insights into the behavior of risk measures under different allocation constraints. 
    What is the relationship among three variations of inf-convolution?
    \item  Question (a) concerns the counter-monotonic inf-convolution of concave or convex distortion risk measures. Is there an explicit formula for the counter-monotonic inf-convolution if the underlying risk measures are neither concave nor convex?
    \item  In the application of portfolio optimization problems, optimal asset allocations often have different structures for agents with varying risk attitudes. How will the optimal strategies change as agents' risk preferences vary, for example, becoming more risk-seeking? Additionally, how does the number of agents involved in the pool affect the optimal allocations?
\end{enumerate}

The paper is dedicated to answering the above four questions by offering general results on counter-monotonic inf-convolution. 
The rest of the paper is organized as follows. Sections \ref{sec:pre} and  \ref{sec:counter} contain preliminaries on risk measures and on risk sharing problems, respectively. In particular, Section \ref{sec:counter} introduces the new concept of counter-monotonic inf-convolution, along with some related discussions on negative dependence. In Sections \ref{sec:concave} and \ref{sec:convex}, we analyze the counter-monotonic risk sharing problem for risk-averse agents and risk-seeking agents, respectively. 
For risk-averse agents, all three forms of inf-convolution lead to the same 
optimal value (Theorem \ref{th:three_h}).
We also characterize conditions for the distortion function that yields equality between the value of the  inf-convolution and that of the risk measure (Theorems  \ref{theorem:counter_two}  and \ref{th:h_n}).
Based on the counter-monotonic improvement theorem in the literature (reported as Theorem \ref{theorem:counter_impro}),
explicit formulas for counter-monotonic inf-convolutions and optimal allocations are obtained for risk-seeking agents (Theorem \ref{theorem:n_convex}). 
Applying these results, we solve the portfolio optimization problem, demonstrating how decision-making varies with agents' risk attitudes, as detailed in Section \ref{sec:port}. In Section \ref{sec:inverse}, we consider agents with inverse S-shaped distortion functions, who are neither risk-averse nor risk-seeking, and obtain an explicit formula for the corresponding counter-monotonic inf-convolution (Theorem \ref{theorem:concave_convex}) under mild conditions.
Section \ref{sec:conclusion} concludes the paper.

\section{Preliminaries}\label{sec:pre}
\subsection{Risk measures}
 Fix  an atomless probability space $(\Omega, \mathcal{F}, \mathbb{P})$. Let $\mathcal{X}$ be a convex cone of random variables  on $(\Omega, \mathcal{F}, \mathbb{P})$, which will be specified in subsequent sections to several different choices.  
 For example, $\X$ may be the set $L^1$  of integrable random variables,
 the set $L^\infty$  of bounded random variables,
 or the set $L^+$ of nonnegative bounded random variables.
 Almost surely equal random variables are treated as identical. 
 We denote by $\mathbb{1}_A$  the indicator function for an event $A \in \mathcal{F}$. A risk measure is a mapping $\rho: \mathcal{X} \rightarrow[-\infty, \infty]$. 
Below we collect some standard properties for a risk measure $\rho$. For any $X, Y \in \mathcal{X}$, 
\begin{enumerate}[(a)]
    \item Monotonicity:  $\rho(X) \leqslant \rho(Y)$ if $X \leqslant Y$;
    \item Law-invariance:  $\rho(X)=\rho(Y)$ if $X$ and $Y$ have the same distribution, i.e., $X \stackrel{\mathrm{d}}{=} Y$; 
    \item  Positive homogeneity: $\rho(\lambda X)=\lambda \rho(X)$ for any $\lambda>0$;
    \item Translation invariance: $\rho (X + c) = \rho (X) + c $ for $c\in \mathbb{R}$, and $X+c\in \X$;
    \item Uniform continuity:\footnote{Continuity of $\rho$ is defined with respect to sup-norm, i.e. $\|X\|=\esssup (|X|)$ for $X \in \mathcal{X}$.}  for all $\varepsilon>0$ there exists $\delta>0$ such that for all $X, Y \in \mathcal{X}, \| X-Y|| \leqslant \delta$ implies $|\rho(X)-\rho(Y)| \leqslant \varepsilon$;
    \item Subadditivity: $\rho(X+Y) \leqslant \rho(X)+\rho(Y)$;
    \item Comonotonic additivity: $\rho(X+Y)=\rho(X)+\rho(Y)$ if $X$ and $Y$ are comonotonic;
    \item Convex order consistency: $\rho(X)\le \rho(Y)$ if $X\le_{\rm cx} Y$, where the inequality is the convex order, meaning $\E[u(X)]\le \E[u(Y)]$ for all convex functions $u$ such that the two expectations are well-defined.
\end{enumerate} 
A coherent risk measure (\cite{artzner1999coherent}) is one that satisfies (a), (c), (d) and (f). 

In (f) and (g), we require that $\rho$ do not take both values $-\infty$ and $\infty$. 
We allow for both $\infty$ and $-\infty$ in the range of risk measures because
we would like to treat inf-convolutions defined below as risk measures,
and  in some cases they can take infinite values.

We consider  $n$ agents, where $n$ is a positive integer, and denote by $[n]=\{1,\dots,n\}$.  
A random vector $(X, Y)$ is said to be comonotonic if
$(X(\omega)-X(\omega^{\prime}))(Y(\omega)-Y(\omega^{\prime})) \geqslant 0 $  for all $\omega, \omega^{\prime} \in \Omega$. 
Comonotonicity of $(X_{1},\ldots, X_{n})$ is equivalent to the existence of increasing functions $f_{i}:\mathbb{R}\mapsto \mathbb{R}$, $i\in[n]$, such that $X_{i}=f_{i}(\sum_{i=1}^nX_i )$
for $i\in [n]$. Terms like ``increasing" or ``decreasing" are in the non-strict sense. We refer to  \cite{dhaene2002concept} for an overview on comonotonicity. 
A pair $(X, Y)$ of random variables  is said to be counter-monotonic if $(X, -Y)$ is comonotonic, i.e., $(X(\omega)-X(\omega^{\prime}))(Y(\omega)-Y(\omega^{\prime})) \leqslant 0 $  for all $\omega, \omega^{\prime} \in \Omega$. 
A random vector $(X_{1},\ldots, X_{n})$ is comonotonic (resp.~counter-monotonic) if each pair of its components is comonotonic (resp.~counter-monotonic). 
Although comonotonicity for $n\ge 3$ is a straightforward extension of the case $n=2$ and well understood,
counter-monotonicity for $n\ge 3$ is highly restrictive on the marginal distributions, and quite different from the case $n=2$.
A stochastic representation of  counter-monotonicity for $n\ge 3$ is provided in \cite{lauzier2023pairwise}, which we present in Proposition \ref{prop:counter} below. 
To emphasize this difference, we sometimes mention ``pairwise counter-monotonicity" for counter-monotonicity in dimension $3$ or higher.

We will model agents' preferences by 
the class of distortion risk measures.
Equivalently, these agents are modelled by the dual utility of \cite{yaari1987dual}, with minimization problems switched to maximization problems. 
To analyze the risk sharing problems, it will be convenient to work with 
a more general class than distortion risk measures, called \emph{distortion riskmetrics} by \cite{wang2020distortion}, which we explain below.
Let $$\H^{\rm BV}=\{h:  [0,1]\to \R \mid  h \,\text{is of bounded variation } h(0)=0 \}.$$
For $h\in \H^{\rm BV} $,
we first define the Choquet integral for a random variable $X$ by
\begin{align*}
\int X \d\left( h\circ \p \right)=\int_0^\infty h(\p(X > x))\d x + \int_{-\infty}^{0} (h(\p(X > x))-h(1) )\d x, 
\end{align*} 
provided that the above is well-defined. 
A distortion riskmetric $\rho_h: \mathcal{X} \rightarrow\mathbb{R}$ is $\rho_h(X)=\int X \d\left( h\circ \p \right)$, where $\X$ is such that the Choquet integral is well-defined. 
Elements of $\H^{\rm BV}$ are called distortion functions, and they are not necessarily monotone.
Denote by  $$\H=\{h:  [0,1]\to [0,1]\mid  h \,\text{is increasing and } h(0)=1-h(1)=0 \},$$ 
which is a subset of $\H^{\rm BV}$. 
In case of $h\in\H $ , the distortion riskmetric $\rho_h$ is a distortion risk measure, which satisfies properties of law-invariance, monotonicity, translation invariance,   positive homogeneity, and comonotonic additivity. 

Characterization and various properties of distortion riskmetrics have been studied on $L^{\infty}$ by \cite{wang2020characterization} and on more general spaces by \cite{wang2020distortion}.
The following are equivalent for $\rho_h$ (see \citet[Theorem 3]{wang2020characterization}): (i) $h$ is concave; (ii) $\rho_h$ is subadditive; (iii) $\rho_h$ is convex; (iv) $\rho_h$ is convex order consistent. 
Moreover, comonotonic additivity and law-invariance plus some continuity characterize the class of $\rho_h$ for $h\in \H^{\rm BV}$.
Many popular risk measures belong to the family of distortion risk measures, 
including the 
regulatory risk measures in banking and insurance, 
the Value-at-Risk (VaR)
and the Expected Shortfall (ES, also known as CVaR and TVaR). 
 
Below we define VaR and ES, with a slight generalization on the domain of the VaR parameter, as done in \cite{embrechts2018quantile}. This generalization of the domain allows some results to be more concise.  
For a random variable $X$, VaR   at level $\alpha \in \mathbb{R}^{+}:=[0, \infty)$ is defined as 
\begin{align}\label{eq:varrr}
    \operatorname{VaR}_\alpha(X)=\inf \{x \in[-\infty, \infty]: \mathbb{P}(X \leqslant x) \geqslant 1-\alpha\} ,
\end{align}
and   ES  at level $\beta \in \mathbb{R}^{+}:=[0, 1)$ is defined as
\begin{align*}
    \operatorname{ES}_\beta(X)=\frac{1}{\beta} \int_0^\beta \operatorname{VaR}_\gamma(X) \mathrm{d} \gamma,
\end{align*}
where $\operatorname{VaR}_\gamma$ is defined in (\ref{eq:varrr}).
Here we use the convention of ``small $\alpha$" as in \cite{embrechts2018quantile}.
If $\alpha \in [0,1)$, 
  $\operatorname{VaR}_{\alpha}$ and $\operatorname{ES}_{\alpha}$ 
are distortion risk measures. They are well-defined on the set of all random variables, and they 
are associated with the distortion functions $h(t)=\id_{\{t>\alpha\}}$ and $ h(t)=\min \{t/\alpha, 1\}$, respectively.

\subsection{Risk sharing and inf-convolution}
\noindent In risk management and game theory, the concept of risk sharing, often referred to as risk allocation, involves distributing the aggregate risk or wealth among multiple agents. A common approach to optimally distribute the risk is by minimizing the overall value of the aggregate risk. 
We assume there are $n$ agents sharing a total loss $X \in \mathcal{X}$ in the market. Suppose that agent $i \in[n]$ has a risk preference modelled by a  risk measure $\rho_{i}$. Given $X \in \mathcal{X}$, we define the set of allocations of $X$ as
\begin{align}\label{def:allo}
    \mathbb{A}_n(X)=\left\{\left(X_1, \dots, X_n\right) \in \mathcal{X}^n: \sum_{i=1}^n X_i=X\right\}.
\end{align}
The allocation set consists of all potential ways to distribute the total risk $X$ among $n$ agents and the associated aggregate risk value is $\sum_{i=1}^n \rho_{i}(X_i)$. 
Note that the definition of allocations crucially depends on the specification of $\X$, which will vary across different applications in the later sections. 

Using \eqref{def:allo}, the inf-convolution $\square_{i=1}^n \rho_{i}$ of $n$   risk measures $\rho_{1}, \dots, \rho_{n}$ is defined as
\begin{align}
    \label{eq:inf}
    \dsquare_{i=1}^n \rho_{i}(X):=\inf \left\{\sum_{i=1}^n \rho_{i}\left(X_i\right):\left(X_1, \dots, X_n\right) \in \mathbb{A}_n(X)\right\}, \quad X \in \mathcal{X}.
\end{align}
That is, the inf-convolution of $n$ risk measures is the infimum over aggregate risk values for all possible allocations.
Like $\mathbb A_n(X)$, 
we make the reliance of $\dsquare_{i=1}^n \rho_{i}$ on $\X$ implicit, but it is useful to keep in mind that $\X$ matters in this definition; that is, for $\X=L^\infty$ and $\X=L^+$, the inf-convolution may differ for the same set of risk measures well-defined on $L^\infty$.

An allocation $(X_1, \dots, X_n)$ is sum-optimal in $\mathbb{A}_n(X)$ if $\square_{i=1}^n \rho_{i}(X)=\sum_{i=1}^n \rho_{i} (X_i)$, i.e., it attains the best total risk value. An allocation $ (X_1, \dots, X_n ) \in \mathbb{A}_n(X)$ is Pareto optimal in $\mathbb{A}_n(X)$ if for any $(Y_1, \dots, Y_n ) \in \mathbb{A}_n(X)$ satisfying $\rho_{i}(Y_i) \leq \rho_{i}(X_i)$ for all $i \in[n]$, we have $\rho_{i}\left(Y_i\right)=\rho_{i}\left(X_i\right)$ for all $i \in[n]$. Pareto optimality means that the allocation cannot be improved for all agents with one agent being strictly improved.

It is well-known that for finite monetary risk measures, which satisfies monotonicity and translation-invariance, Pareto optimality is equivalent to sum-optimality; see  \citet[Proposition 1]{embrechts2018quantile}.
Since we focus on monetary risk measures (particularly, distortion risk measures) in this paper,
we will simply say that an allocation is optimal if sum-optimality holds. 

\section{Comonotonic and counter-monotonic risk sharing}
\label{sec:counter}

As  in other multivariate models in risk management, the 
dependence structure among the components of the risk allocation
is important for interpreting the economic incentives created by the allocation. 
For instance, if the allocation is comonotonic, then all agents receive gains and losses together; if the allocation is counter-monotonic, then 
all agents are essentially gambling against each other.
In some situations, e.g., in an insurance setting, it may be preferred or mandatory to allocate the aggregate risk in a comonotonic way among agents, who are insureds and insurers in this context. 
On the other hand, in a lottery setting, the allocation of payoffs may be counter-monotonic. 

We will consider risk sharing problems constrained to comonotonic or counter-monotonic allocations.
The comonotonic risk sharing problem is well studied in the risk management and insurance literature; see e.g., \cite{jouini2008optimal}, \cite{cui2013optimal}, and \cite{boonen2021competitive}. 
The counter-monotonic problem would be the focus of our paper, though we will compare these three types of risk sharing problems throughout. The set of comonotonic allocations is defined as
\begin{align*}
    \mathbb{A}_n^{+}(X)=\left\{\left(X_1, \ldots, X_n\right) \in \mathbb{A}_n(X): X_1, \ldots, X_n\right. \text{are comonotonic} \}.
\end{align*} 
The corresponding comonotonic inf-convolution $\boxplus_{i=1}^n \rho_{i}$ of risk measures $\rho_{1}, \dots, \rho_{n}$ is defined as
$$
\underset{i=1}{\stackrel{n}{\boxplus}} \rho_{i}(X):=\inf \left\{\sum_{i=1}^n \rho_{i}\left(X_i\right):\left(X_1, \ldots, X_n\right) \in \mathbb{A}_n^{+}(X)\right\}.
$$
An allocation $\left(X_1, \ldots, X_n\right) \in \mathbb{A}_n^{+}(X)$ is called an optimal allocation of $X$ for $\left(\rho_{1}, \ldots, \rho_{n}\right)$ within $ \mathbb{A}_n^{+}$ 
if $\sum_{i=1}^n \rho_{i}\left(X_i\right)=\boxplus_{i=1}^n \rho_{i}(X)$.  
By definition, it is clear that $\square_{i=1}^n \rho_{i}(X) \leqslant \boxplus_{i=1}^{n}\rho_{i}(X)$. Hence, if an optimal allocation of $X$ is comonotonic, then it is also an optimal   allocation within  $ \mathbb{A}_n^{+}$, and further we have $\square_{i=1}^n \rho_{i}(X)=\boxplus_{i=1}^n \rho_{i}(X)$. For law-invariant convex risk measures on $L^{\infty}$, optimal constrained allocations are also optimal allocations; see \cite{jouini2008optimal}. This statement remains true if the underlying risk measures preserve convex order, and this is based on the comonotone improvement of \cite{landsberger1994co}. 

Our new invention is the risk sharing framework where allocations are restricted in the set of counter-monotonic   allocations, which can be rigorously formulated  below, by
\begin{align*}
    \mathbb{A}_{n}^{-}(X) = \left\{(X_{1}, \ldots, X_{n})\in \mathbb{A}_{n}(X): X_{1}, \ldots, X_{n} \, \text{are counter-monotonic} \right\}.
\end{align*}
The corresponding counter-monotonic inf-convolution $ \dboxminus_{i=1}^{n} \rho_{i}$ is thus defined as 
\begin{align*}
    \dboxminus_{i=1}^{n} \rho_{i}(X) = \inf \left\{\sum_{i=1}^{n}\rho_{i}(X_{i}): (X_{1}, \ldots, X_{n}) \in \mathbb{A}_{n}^{-}(X) \right\}.
\end{align*} 
Similarly, an allocation $\left(X_1, \ldots, X_n\right) \in \mathbb{A}_n^{-}(X)$ is called an optimal  allocation of $X$ within $\mathbb A^{-}_n$ if $\sum_{i=1}^n \rho_{i}\left(X_i\right)=\dboxminus_{i=1}^n \rho_{i}(X)$. It is clear that $\square_{i=1}^n \rho_{i}(X) \leqslant \dboxminus_{i=1}^{n}\rho_{i}(X)$. In contrast to the rich literature on comonotonic risk sharing, research on counter-monotonic risk sharing problem is quite limited. Counter-monotonic allocations  have been recently explored by \cite{lauzier2024negatively}, but with a different approach than ours. More precisely, instead of working with  
$ \mathbb{A}_{n}^{-}(X)$ directly, \cite{lauzier2024negatively} considered conditions under which the original problem within $\mathbb A_n$ has an optimal solution within $ \mathbb{A}_{n}^{-}$.

Before we analyze the counter-monotonic risk sharing problem, we first recall the stochastic representation of counter-monotonicity given by \cite{lauzier2024negatively}, which will be useful throughout our analysis. 
\begin{proposition}[\cite{lauzier2024negatively}]\label{prop:counter}
    For $X \in \mathcal{X}$ and $n\ge 3$, suppose that at least three of $\left(X_1, \ldots, X_n\right) \in \mathbb{A}_n(X)$ are non-degenerate. Then $\left(X_1, \ldots, X_n\right)$ is counter-monotonic if and only if there exist constants $m_1, \ldots, m_n$ and $\left(A_1, \ldots, A_n\right) \in \Pi_n$ such that
\begin{align}\label{eq:counter_form1}
    X_i=(X-m) \mathbb{1}_{A_i}+m_i \quad \text { for all } i \in[n] \text { with } m=\sum_{i=1}^n m_i \leq \operatorname{ess-inf} X
\end{align}
or
\begin{align}\label{eq:counter_form2}
    X_i=(X-m) \mathbb{1}_{A_i}+m_i \quad \text { for all } i \in[n] \text { with } m=\sum_{i=1}^n m_i \geq \operatorname{ess-sup} X.
\end{align}
\end{proposition}

By taking $m_{1}=\ldots=m_{n}=0$, i.e., $X \geq 0$ or $X \leq 0$, a simple counter-monotonic allocation in the form of (\ref{eq:counter_form1}) and (\ref{eq:counter_form2}) is given by
\begin{align*}
    X_i=X \mathbb{1}_{A_i} \, \text{for each} \, i\in [n] , \, \text{where}\, \left(A_1, \ldots, A_n\right) \in \Pi_n
\end{align*}
Specifically, such an allocation is called a \textit{jackpot allocation} if $X \geq 0$ and a \textit{scapegoat allocation} if $X \leq 0$ by \cite{lauzier2024negatively}.
It is clear to see that there is a ``winner-takes-all" structure in a jackpot allocation and a ``loser-loses-all"  structure in a scapegoat allocation.
In real life, 
such $\left(X_1, \ldots, X_n\right)$ may represent the outcome of $n$ lottery tickets, exactly one of which wins a random reward of $X$. ``Roulette wheel decisions''\footnote{For tasks no one wants to do, a roulette wheel with everyone's names on it can be spun to determine who has to undertake the task.} is an example for scapegoat allocations.

A special form of jackpot and scapegoat allocations will be useful in our analysis. 
The allocation $(X_1,\dots,X_n)$
of $X$ with
 either $X\ge 0$ or $X\le 0$ 
is called a \emph{uniform counter-monotonic allocation} if 
$X_{i}=X\mathbb{1}_{A_{i}}$ with $\mathbb{P}(A_{i})=1/n$ for each $i\in[n]$, where $(A_{1}, \ldots, A_{n})\in \Pi_{n}$ is independent of $X$. 
For any $X\ge 0$ or $X\le0$, a uniform counter-monotonic alloction exists as soon as there exists a (nondegenerate) uniform random variable independent of $X$.

Both the set of comonotonic allocations $\mathbb{A}_{n}^{+}(X)$ and the set of counter-monotonic allocations $\mathbb{A}_{n}^{-}(X)$ are strict subsets of the set of all potential allocations $\mathbb{A}_{n}(X)$. Hence, the sequel refers to the problem of sharing risk in 
$\mathbb{A}_{n}(X)$, 
$\mathbb{A}_{n}^{+}(X)$ and $\mathbb{A}_{n}^{-}(X)$ as unconstrained,  comonotonic and counter-monotonic risk sharing, respectively. 
A first and natural question is whether 
$\dboxplus_{i=1}^n \rho_{i}(X) $
and $ \dboxminus_{i=1}^{n}\rho_{i}(X)$ admit a clear relation in terms of their values. 
In the next few sections, we will answer this question and several others mentioned in the Introduction. 
 
Although we have introduced the framework of counter-monotonic risk sharing with possibly different risk measures for each agent, 
in the rest of the paper, we will always consider the case of homogeneous agents; that is, all agents have the same risk measure, and it is  a distortion risk measure. 
This homogeneous setting, although being restrictive, already gives rise to interesting mathematical results. 
We leave the heterogeneous case, as well as other decision models, for future study.

\section{Risk-averse agents and pseudo-concave capacities}\label{sec:concave}

In this section, we take $\mathcal{X}$ as the set $L^\infty$ of bounded random variables.
Assume that all agents use the same distortion risk measure $\rho_h$. We first present a general result that holds for all types of distortion functions $h$, in which
a clear ordering among the unconstrained, comonotonic and counter-monotonic inf-convolutions is provided, and they become equivalent when $h$ is concave.

\begin{theorem}\label{th:three_h}
    For $h\in\mathcal{H}$, the following  always holds:    \begin{align}\label{ineq:three_conv}
        \dsquare_{i=1}^{n}\rho_{h} \leq \dboxminus_{i=1}^{n} \rho_{h} \leq \dboxplus_{i=1}^{n}\rho_{h} = \rho_{h}.
    \end{align}
    Moreover, if $h$ is concave, then 
    \begin{align}\label{ineq:three_conv2} \dsquare_{i=1}^{n}\rho_{h} = \dboxminus_{i=1}^{n} \rho_{h}  = \dboxplus_{i=1}^{n}\rho_{h} = \rho_{h}.
    \end{align}
    
\end{theorem}
A proof of Theorem \ref{th:three_h} is straightforward. 
For the first inequality,  
we have seen above that the unconstrained inf-convolution is the smallest since there are no constraints on the allocation set. This relation also holds for agents with other risk preferences than distortion risk measures.
The second inequality holds by taking allocation $(X_{1}, \dots, X_{n})=(X, 0, \ldots, 0)$, that is both counter-monotonic and comonotonic, thus leading to $\boxminus_{i=1}^{n} \rho_{h} \leq \rho_{h}$
and $\boxplus_{i=1}^{n} \rho_{h} \leq \rho_{h}$.
The last equality in \eqref{ineq:three_conv} follows directly from comonotonic additivity of 
$\rho_h$.  
To show \eqref{ineq:three_conv2}, it suffices to note that 
$\dsquare_{i=1}^{n}\rho_{h}  = \dboxplus_{i=1}^{n}\rho_{h} $
for concave $h$, which follows from the  comonotonic improvement of \cite{landsberger1994co}; see also Proposition 5 of \cite{embrechts2018quantile}. 

The general chains of relations in \eqref{ineq:three_conv} 
and \eqref{ineq:three_conv2}
 do not necessarily generalize to  risk measures that are not comonotonic additive,   or to heterogeneous risk preferences among agents.  

In the following example, we compare the three inf-convolutions in (\ref{ineq:three_conv}) for risk measures being VaR.
\begin{example}[VaR]\label{ex:var}
    We now analyze the relationship among counter-monotonic, comonotonic and  unconstrained optimal allocations in the risk sharing problem for VaR agents with levels $\alpha_{1}, \dots, \alpha_{n}, \alpha \in (0,1)$. Here, we assume $\alpha_{i}=\alpha$ for each $i\in [n]$. 
    In particular, the values of $\dsquare_{i=1}^n \VaR_{\alpha}(X)$ and $\dboxplus_{i=1}^n \VaR_{\alpha}(X)$ are given by Theorem 2 and Proposition 5 of \cite{embrechts2018quantile}. They yield \begin{align}\label{eq:var}
        \dsquare_{i=1}^n \VaR_{\alpha}(X)
        =\operatorname{VaR}_{n \alpha}(X) \leq \dboxplus_{i=1}^n \VaR_{\alpha}(X) = \operatorname{VaR}_{\alpha}(X) .
    \end{align}
    Moreover, if $\alpha < 1/n$, by \citet[Theorem 2]{embrechts2018quantile},  there exists a Pareto-optimal allocation $(X_{1}, \ldots, X_{n})$ of $X$ with the form of 
    \begin{align*}
        X_i=(X-m) \mathbb{1}_{A_i}, i \in[n-1] \quad \text { and } \quad X_n=(X-m) \mathbb{1}_{A_n}+m
    \end{align*}
    for some $(A_1, \ldots, A_n)\in \Pi_{n}$, 
    where $ m\in (-\infty, \operatorname{VaR}_{n \alpha}(X)]$ is a constant.  By taking $ m = \essinf X$, $(X_{1}, \ldots, X_{n})$  is counter-monotonic, as noticed by 
    \citet[Theorem 1]{lauzier2023pairwise}. Hence, $\boxminus_{i=1}^n \VaR_{\alpha}(X) = \dsquare_{i=1}^n \VaR_{\alpha}(X)$, and  thus the inf-convolution of several VaRs is equivalent to the counter-monotonic one. 
    Therefore, for any $X \in \mathcal{X}$, the following always holds true
    \begin{align}\label{ineq:relation_var}
        \dboxminus_{i=1}^n \VaR_{\alpha}(X)= \dsquare_{i=1}^n \VaR_{\alpha}(X) \leq \dboxplus_{i=1}^n \VaR_{\alpha}(X)
    \end{align}
    and generally the inequality in (\ref{ineq:relation_var}) is not an equality.
    Furthermore, we can analogously show that the result (\ref{ineq:relation_var}) also holds true for RVaR agents.  
    It is notable that the distortion risk function corresponding to VaR is neither concave nor convex, which implies that concavity or convexity of $h$ is not necessary for $\dboxminus_{i=1}^n \rho_{h}$ being equivalent to $\dsquare_{i=1}^n \rho_{h}$,
    although concavity is sufficient as we see in \eqref{ineq:three_conv2}.
\end{example}

\begin{example}[ES] 
It is well known that ES has a concave distortion function, and hence,    by Theorem \ref{th:three_h}, 
$
      \dsquare_{i=1}^n \ES_{\beta} =   \dboxminus_{i=1}^n \ES_{\beta}= 
 \dboxplus_{i=1}^n \ES_{\beta} = \ES_{\beta}. $ 
\end{example}

Without assuming that $h$ is concave, 
solving the counter-monotonic inf-convolution poses a challenge, as is the case of the unconstrained  problem. 
Following Theorem 1 of \cite{aouani2021propensity}, which shows that  counter-monotonic supperadditive Choquet functional can be characterized by pseudo-convex capacity, we provide some necessary and sufficient conditions for the equality $\boxminus_{i=1}^{n} \rho_{h}= \rho_{h}$.
We first review some useful concepts in the theory of capacity.  

A set function $\nu: \mathcal{F}\mapsto \mathbb{R}$ is called a \emph{capacity} if $\nu(\emptyset)=0$ and $\nu$ is \emph{monotone}, i.e., $A \subseteq B \Rightarrow v(A) \leq v(B)$. 
Moreover,
$\nu$ is \emph{normalized} if $\nu(\Omega)=1$.
In particular, a capacity $\nu$ on $\mathcal{F}$ is called 
\begin{itemize}
        \item[(i)] \emph{concave} if \begin{equation}\label{eq:concave}\nu(A \cup B) + \nu(A \cap B) \leq  \nu(A) + \nu(B)
        \end{equation} 
        holds for all $A, B \in \mathcal{F}$;
        \item[(ii)] \emph{subadditive} if \eqref{eq:concave} holds for all disjoint $A, B \in \mathcal{F}$;

    \item[(iii)] 
    \emph{concave at the sure event} if \eqref{eq:concave} holds for all $A, B \in \mathcal{F}$ with $A \cup B = \Omega$;
    \item[(iv)]  \emph{pseudo-concave} if it is subadditive and concave at the sure event.
\end{itemize}
In the literature, concavity (convexity) of $\nu$ is also called submodularity (supremodularity). 
Clearly, concavity is stronger than the other three properties. 

Let  $\nu$ be a capacity.
The Choquet integral of a bounded random variable $X$  with respect to $\nu$ is defined as 
$$
\int X \mathrm{~d} \nu=\int_0^{\infty} \nu(X>x) \mathrm{d} x+\int_{-\infty}^0(\nu(X>x)-\nu(\Omega)) \mathrm{d} x
$$ 
The distortion risk measure $\rho_h$ is an increasing Choquet integral with $\nu = h\circ \mathbb{P}$ where $\nu$ is normalized. The properties of $\nu$ are  closely related to those of $h$. For instance, concavity of $\nu $ is equivalent to concavity of $h$; see \cite{follmer2011stochastic} or \cite{marinacci2004introduction}.
Next, we introduce a property for $h\in \mathcal H$, which we refer to as dual subadditivity, that turns out to correspond to pseudo-concavity of $\nu$.

\begin{definition}
Let $h\in \mathcal H$.
\begin{enumerate}[(i)]
    \item The \emph{dual function} of $h$ is defined by 
    $\tilde h(t) =1-h(1-t)$ for $t\in [0,1]$. 
    It is clear that $\tilde h
    \in \mathcal H$. 
    \item  The function
    $h$ is \emph{dually subadditive} if  $h$  satisfies subadditivity (i.e., $h(x+y)\leq h(x) + h(y)$ for $x, y\in [0,1]$ with $x+y \leq 1$)
    and  $\tilde{h}(x)$  
    is superadditive
    (i.e., $\tilde{h}(x+y)\geq \tilde{h}(x) + \tilde{h}(y)$ for $x, y\in [0,1]$ with $x+y \leq 1$).
    \end{enumerate}
\end{definition}
Note that $h$ is concave if and only $\tilde h$ is convex.
It is straightforward to check that any concave function $h\in \mathcal H$ is dually subadditive.  
The following theorem gives a necessary and sufficient condition for  $\rho_h \boxminus \rho_h=\rho_h$.

\begin{theorem}\label{theorem:counter_two}
    Suppose $h\in \mathcal{H}$ and denote by $\nu = h \circ \mathbb{P}$. 
    The following are equivalent.
    \begin{itemize}
        \item[(i)]  $\rho_{h} \boxminus \rho_{h}(X) = \rho_{h}(X)$ holds for all $X\in \mathcal{X}$.
        \item[(ii)] $h$ satisfies dual subadditivity.
        \item[(iii)] $\nu$ is pseudo-concave. 
    \end{itemize}
     
\end{theorem}

\begin{proof}
    We first show the implication (i) $\Rightarrow$ (ii). Clearly, statement (i) implies $\rho_{h}(X) \leq \rho_{h}(X_{1}) + \rho_{h}(X_{2})$ for $(X_{1}, X_{2})\in \mathbb{A}_{2}^{-}(X)$.
    Let  $A, B \in \mathcal{F}$ be such that $ A\cap B =\emptyset$. It is straightforward to show $h$ is subadditive
    by taking $X_{1}= \mathbb{1}_{A}$ and $X_{2}= \mathbb{1}_{B}$.
    Moreover, for such $A$ and $B$, it can be verified that $X_{1} = \mathbb{1}_{A \cup B}$ and $X_{2}=\mathbb{1}_{A^\mathsf{c}} $ are counter-monotonic (see Table \ref{tab:examp_counter}). It follows that $h(x+y-1)+1 \leq h(x)+h(y)$ for $x, y \in [0,1]$ with $x+y\geq 1$, which is equivalent to superadditivity of $\tilde{h}$ by some simple manipulations. 

    \begin{table}
        \centering
        \begin{tabular}{cccc}\hline
            & $A$ & $B$  & $A^\mathsf{c}\cap B^\mathsf{c}$  \\ \hline
       $\mathbb{1}_{A \cup B} $ & 1 & 1 & 0 \\ 
         $\mathbb{1}_{A^\mathsf{c}} $ & 0  & 1 & 1  \\ 
         $\mathbb{1}_{A \cup B} + \mathbb{1}_{A^\mathsf{c}} $ & 1  & 2 & 1 \\ \hline
        \end{tabular}
        \caption{A counter-monotonic allocation}
        \label{tab:examp_counter}
    \end{table}
       
    Next, we show (ii) $\Rightarrow$ (iii). Statement (ii) can be rewritten as $h$ being subadditive and satisfying $h(x+y-1)+1 \leq h(x)+h(y)$ for $x, y \in [0,1]$ with $x+y\geq 1$, which is equivalent to pseudo-concavity of $\nu$; see \citet[Example 14]{principi2023antimonotonicity}.

For a proof of (iii) $\Rightarrow$ (i), first note that $\rho_{h} \boxminus \rho_{h}(X) \leq \rho_{h}(X)$ by Theorem \ref{th:three_h}. Thus it suffices to show $\rho_{h}$ is counter-monotonic subadditive,
which directly follows Theorem 1 of \cite{aouani2021propensity}. 
\end{proof}


As we know, pairwise counter-monotonicity is the generalization of counter-monotonicity for the case of $ n\geq 3$. 
The next result shows that Theorem \ref{theorem:counter_two} can be generalized to the case of $n$ agents with identical distortion risk functions. 

\begin{theorem}\label{th:h_n}
Let $h\in\mathcal{H}$ and $n\ge 3$. 
\begin{enumerate}[(i)]
\item  
$ \dboxminus_{i=1}^{n} \rho_{h}= \rho_{h}$ holds true if and only if $h$ is dually subadditive.
\item  
$ \dsquare_{i=1}^{n} \rho_{h}= \rho_{h}$ holds true if and only if $h$ is concave.
\item  $ \dboxplus_{i=1}^{n} \rho_{h}= \rho_{h}$ holds true for all $h$.
\end{enumerate}
\end{theorem}
\begin{proof}
(i) 
It is straightforward to show the ``only if'' part by Theorem \ref{theorem:counter_two}. Next, we demonstrate the ``if'' statement.
To prove whether the theorem can be generalized to $n$ agents for $n \geq 3$, it suffices to show whether the inf-convolution of $n$ risk measures can be understood as applying the inf-convolutions of two risk measures repeatedly. Specifically, we aim to show $\dboxminus_{i=1}^{n}\rho_{h}=  {\rho_{h}\dboxminus \dots \dboxminus \rho_{h}} $, where the right-hand side has $n$ terms of $\rho_h$.
For simplicity, we denote the repeated inf-convolution of $n$ risk measures from left to right as  $\rho_{h}\dboxminus \dots \dboxminus \rho_{h}$. 

For a given $X\in \mathcal{X}$, we define a new set of allocations of $X$, that is,
$$
\mathbb{A}^{*}_{n}(X)=\left\{(X_{1}, \dots, X_{n})\in \mathbb{A}_{n}(X): 
\sum_{i=1}^{j-1}X_{i}\ \text{and}\  X_{j} \ \text{are counter-monotonic for}\ j\in[n]\right\}.
$$
Then $\rho_{h}\dboxminus \dots \dboxminus \rho_{h}$ can be expressed as: 
\begin{align}
    \rho_{h}\dboxminus \dots \dboxminus \rho_{h}(X)=\inf \left\{\sum_{i=1}^{n} \rho_{h}(X_{i}): (X_{1}, \dots, X_{n})\in \mathbb{A}^{*}_{n}(X)
    \right\}.
\end{align}
It is straightforward to show $\rho_{h}\dboxminus \dots \dboxminus \rho_{h} \leq \dboxminus_{i=1}^{n}\rho_{h}$. 
In fact, the inequality also holds true when applying multiple $\rho_{i}$ for $i\in [n]$. 
The result directly follows part (ii) of Theorem 2 in \cite{lauzier2023pairwise}, implying 
$\mathbb{A}^{*}_{n}(X)$ is a subset of $\mathbb{A}_{n}(X)$. 

Next, we show the converse direction. 
For any $(X_{1}, \dots, X_{n})\in \mathbb{A}^{*}_{n}(X)$, by repeatedly using dually subadditivity of $h$ and Theorem \ref{theorem:counter_two}, we have 
\begin{align*}
    \sum_{i=1}^{n}\rho_{h}(X_{i}) &=\rho_{h}(X_{1})+ \rho_{h}(X_{2})+ \sum_{i=3}^{n}\rho_{h}(X_{i})
    \\& \geq \rho_{h}(X_{1}+X_{2}) + \sum_{i=3}^{n}\rho_{h}(X_{i}) 
    \\ & \geq  \rho_{h}(X_{1}+\dots + X_{n-1}) + \rho_{h}(X_{n})
    \\ & \geq \rho_{h}(X) \geq \dboxminus_{i=1}^{n}\rho_{h}(X), 
\end{align*}
which implies that $ {\rho_{h}\dboxminus \dots \dboxminus \rho_{h}}\geq \dboxminus_{i=1}^{n}\rho_{h}$. 
Combining the above results, we can conclude that $\dboxminus_{i=1}^{n}\rho_{h} = \rho_{h}\dboxminus \dots \dboxminus \rho_{h}= \rho_{h}$.

(ii) The ``if'' part directly follows Theorem \ref{th:three_h}. It is also straightforward to prove the ``only if'' part since concavity of $h$ is equivalent to subadditivity of $\rho_h$.

(iii) The result holds trivially from comonotonic additivity of $\rho_{h}$.
\end{proof}

\begin{remark} The inf-convolution of $n$ risk measures $\dsquare_{i=1}^{n} \rho_{{i}}$ is equal to the repeated application of inf-convolutions of two risk measures; see \citet[Lemma 2]{liu2020}. That is, 
$\dsquare _{i=1}^n \rho_{{i}} = \rho_{{1}} \dsquare  
 \dots \dsquare \rho_{{n}}  .$
 Similarly, it also holds that
 $\dboxplus _{i=1}^n \rho_{{i}} = \rho_{{1}} \dboxplus  
 \dots \dboxplus \rho_{{n}}$. 
The relation does not hold true for the counter-monotonic inf-convolution. A counter-example is given by Example \ref{exa:repeat} in Appendix \ref{sec:counter-example}, showing
$$\dboxminus_{i=1}^n \rho_{{i}}(X) > \rho_{{1}} \boxminus 
 \dots \boxminus \rho_{{n}}(X) \mbox{~~~
 for some $X\in \mathcal X$.}$$
 This is due to the fact that counter-monotonicity behaves quite differently in the cases $n=2$ and $n\ge 3$. 
 Nevertheless, $\dboxminus_{i=1}^n \rho_{{i}}= \rho_{{1}} \boxminus 
 \dots \boxminus \rho_{{n}}$ can be expected in many  relevant cases. 
For instance, it
 holds true for agents with dually subadditive distortion functions, as shown in Theorem \ref{th:h_n}. Additionally, it also holds when risks are measured with VaRs because the unconstrained inf-convolution and counter-monotonic inf-convolution coincide, as shown in Example \ref{ex:var}.
\end{remark}



We can immediately obtain the following corollary of Theorem \ref{th:h_n}.
\begin{corollary}
    Suppose $\mathcal X=L^\infty$.
    If $h\in \mathcal{H}$ is convex and not the identity, then a comonotonic allocation of $X$ is never Pareto optimal.
\end{corollary}

\begin{proof}
    Let $\mathcal X=L^\infty$. We first assume some comonotonic allocations of $X$ are Pareto optimal. Thus, there exists an allocation $(X_{1}, \ldots, X_{n})\in \mathbb{A}_{n}^{+}$ such that 
    \begin{align*}
        \sum_{i=1}^{n} \rho_{h}(X_{i})=\rho_{h}(X)=\dsquare_{i=1}^{n} \rho_{h}(X).
    \end{align*}
    The first equality is from the comonotonic additivity of $\rho_{h}$ and the second is from the Pareto optimality of $(X_{1}, \ldots, X_{n})$. 
    By Theorem \ref{th:h_n}, the second equality implies that $h$ is concave, a contradicting $h$ being convex and not the identity.
    Therefore, a comonotonic allocation $(X_{1}, \ldots, X_{n})\in  \mathbb{A}_{n}^{+}$ cannot be Pareto optimal.
\end{proof}

So far, we have seen that for concave distortion functions, 
the counter-monotonic inf-convolution gives equivalent results to the comonotonic and unconstrained ones. 
The situation is drastically different when risk-seeking agents are involved. 
This will be discussed in the next section.

\section{Risk-seeking agents}\label{sec:convex}
In this section, our aim is to investigate the risk sharing problem for risk-seeking agents and further determine the explicit form of the corresponding counter-monotonic inf-convolutions. The set of random variables $\mathcal{X}$ in this section may be the set $L^\infty$ of bounded random variables, the set $L^+$ of nonnegative bounded random variables, or the set $L^-$ of nonpositive bounded random variables. We will specify which set is used in each case.
As we know, it always holds that $\dboxminus_{i=1}^n \rho_{h}\leq \rho_h$. Counter-monotonic risk sharing among risk-averse agents achieves the upper bound $\rho_h$, as shown in Theorem \ref{th:h_n}.

We now take a small detour to consider sup-convolution instead of inf-convolution. 
The following proposition shows
$\rho_h$
is indeed the 
supremum of aggregate risks for risk-seeking agents over all possible allocations, whether in the general sense or in the counter-monotonic sense.

\begin{proposition}\label{prop:convex_sup}
    Suppose $\mathcal{X}=L^\infty$ and 
    $X\in \mathcal{X}$.
    If $h\in \mathcal{H}$ is convex, then \begin{align*}
         \rho_{h}(X)&=\sup \left\{\sum_{i=1}^{n}\rho_{h}(X_{i}): (X_{1}, \ldots, X_{n})\in \mathbb{A}_{n}(X)\right\}
    \\& =\sup \left\{\sum_{i=1}^{n}\rho_{h}(X_{i}): (X_{1}, \ldots, X_{n})\in \mathbb{A}_{n}^{-}(X)\right\}.\end{align*}
\end{proposition}
\begin{proof}
If $h$ is convex, then $\rho_h$ is superadditive (\citet[Theorem 3]{wang2020characterization}), implying
\begin{align*}
         \rho_{h}(X)&\ge \sup \left\{\sum_{i=1}^{n}\rho_{h}(X_{i}): (X_{1}, \ldots, X_{n})\in \mathbb{A}_{n}(X)\right\}
    \\& \ge \sup \left\{\sum_{i=1}^{n}\rho_{h}(X_{i}): (X_{1}, \ldots, X_{n})\in \mathbb{A}_{n}^{-}(X)\right\}\ge \rho_h(X),\end{align*}
    where the last inequality is due to $(X,0,\dots,0)\in \mathbb A_n^{-}(X)$. 
  This gives the desired equalities. 
\end{proof}

Proposition \ref{prop:convex_sup} shows that $\rho_{h}$ is indeed the sup-convolution of several concave distortion risk measures.  
Unlike the comonotonic risk sharing problem, where we always have $\dboxplus_{i=1}^{n}\rho_h=\rho_h$ for any $h$, regardless of $h$ is concave or convex, the counter-monotonic risk sharing is more complicated. 
In particular, for risk-seeking agents, the equality in $\dboxminus_{i=1}^{n}\rho_{h} \leq \rho_{h}$ generally does not hold true, which can also be verified by some numerical results in Table \ref{tab:examp_counter}.

The technique of counter-monotonic improvement theorem, a converse to  the comonotonic improvements, as introduced in \cite{lauzier2024negatively}, will be helpful for our next results.  

\begin{theorem}[\cite{lauzier2024negatively}]\label{theorem:counter_impro}
     Let $X_1, \ldots, X_n \in L^1$ be nonnegative and $X=\sum_{i=1}^n X_i$. Assume that there exists a uniform random variable $U$ independent of $X$. Then, there exists $\left(Y_1, \ldots, Y_n\right) \in \mathbb{A}_n(X)$ such that (i) $\left(Y_1, \ldots, Y_n\right)$ is counter-monotonic; (ii) $Y_i \geq_{\mathrm{cx}} X_i$ for $i \in[n]$; (iii) $Y_1, \ldots, Y_n$ are nonnegative. Moreover, $\left(Y_1, \ldots, Y_n\right)$ can be chosen as a jackpot allocation.
\end{theorem}

The counter-monotonic improvement theorem indicates the jackpot allocation is always be preferred by risk-seeking agents. To apply the counter-monotonic improvement theorem, it is important to emphasize the technical assumption that there exists a (nondegenerate) uniform random variable $U$  independent of $X$. To formalize this, we introduce the following set:
    $$\mathcal{X}^{\perp}=\{X \in \mathcal{X}: \text{there exists a uniform random variable $U$ independent of} \,X\}.$$

The next lemma provides the optimal value for the inf-convolution when the allocation is constrained to be jackpot allocations, given that the total risk $X$ is nonpositive or nonnegative. The lemma would be helpful to establish our main result (Theorem \ref{theorem:n_convex}).

\begin{lemma}\label{lemma:convex}
   Suppose $X \in \mathcal{X}^{\perp}$ and 
   $h\in \mathcal{H}$ is convex. 
   Denote by  
   \begin{align*}
       I(X) = \inf \left\{\sum_{i=1}^{n}\rho_{h}(X\mathbb{1}_{A_{i}}): (A_{1}, \ldots, A_{n})\in \Pi_{n}\right\}.
   \end{align*}
   The following statements hold.
    \begin{itemize}
        \item[(i)] If $\mathcal X=L^+$, then $I=\rho_{g}$, where $g(t)=n h(t/n)$.
        \item[(ii)] If $\mathcal X=L^-$, then $I=\rho_{g}$, where $g(t)=nh(1-(1-t)/n)-nh(1-1/n)$.                    
    \end{itemize}
\end{lemma}
\begin{proof}
(i) Let $\mathcal X=L^+$ and $X\in \mathcal{X}^{\perp}$. It follows that
 \begin{align*}
        I(X)&= \inf \left\{\sum_{i=1}^{n} \int_{0}^{\infty} h(\mathbb{P}(X \mathbb{1}_{A_{i}}>x)) \mathrm{d} x: (A_{1}, \ldots, A_{n})\in \Pi_{n}\right\}\\ 
        &\geq \int_{0}^{\infty} \inf \left\{\sum_{i=1}^{n}  h(x_{i}) : \sum_{i=1}^{n}x_{i}=\mathbb{P}(X>x)\right\} \mathrm{d} x\\
        &= \int_{0}^{\infty} nh
        \left(\frac{1}{n}\mathbb{P}(X >x)\right) \mathrm{d} x.
    \end{align*}
The last equality holds because of the convexity of $h$. 

Furthermore, 
since we assume that there exists a uniform $U$ independent of $X$, we can always find a partition $(A_{1}, \ldots, A_{n})\in \Pi_{n}$ independent of $X$ such that $\mathbb{P}(A_{i})=1/n$ for all $i\in[n]$. Then we can obtain   
     \begin{align*}
        \sum_{i=1}^{n}\rho_{h}(X\mathbb{1}_{A_{i}})= \int_{0}^{\infty} nh\left(\frac{1}{n}\mathbb{P}(X >x)\right) \mathrm{d} x.
    \end{align*}
    This implies that $I(X)\leq \rho_{g}(X)$, where $g(t)=n h(t/n)$, $t\in [0,1]$.
    Combining the above inequalities, the desired result is obtained.
    
    (ii) Let $\mathcal X=L^-$ and $X\in \mathcal{X}^{\perp}$. It follows that 
    \begin{align*}
        I(X)&=\inf \left\{\sum_{i=1}^{n} \int_{-\infty}^{0} (h(\mathbb{P}(X \mathbb{1}_{A_{i}}>x))-1) \mathrm{d} x: (A_{1}, \ldots, A_{n})\in \Pi_{n}\right\}\\ 
        &= \inf \left\{-\sum_{i=1}^{n} \int_{-\infty}^{0} \tilde{h}(\mathbb{P}(X \mathbb{1}_{A_{i}}\leq x)) \mathrm{d} x: (A_{1}, \ldots, A_{n})\in \Pi_{n}\right\}\\
        & \geq -\int_{-\infty}^{0}\sup \left\{\sum_{i=1}^{n}  \tilde{h}(\mathbb{P}(X \mathbb{1}_{A_{i}}\leq x)): (A_{1}, \ldots, A_{n})\in \Pi_{n}\right\}
        \mathrm{d} x\\
        &= -n \int_{-\infty}^{0} \tilde{h}\left(\frac{1}{n}\mathbb{P}(X \leq x)\right) \mathrm{d} x,
    \end{align*}
    where $\tilde{h}(t) = 1-h(1-t)$. The last equality can be obtained by convexity of $h$, implying $\tilde{h}$ is concave. 
    Furthermore, 
    the above inequality can be rewritten as 
    \begin{align*}
        I(X) &\geq -n \int_{-\infty}^{0} \tilde{h}\left(\frac{1}{n}\mathbb{P}(X \leq x)\right) \mathrm{d} x \\
        &= -n \int_{0}^{\infty} \tilde{h}\left(\frac{1}{n}\mathbb{P}(-X \geq x)\right) \mathrm{d} x = -\rho_{\tilde{h}^{*}}(-X)= \rho_{g}(X),
    \end{align*}
    where $\tilde{h}^{*}(t)=n\tilde{h}(t/n)$. 
    The last equation directly follows from \citet[Lemma 2]{wang2020characterization}, and $g(x)$ is given by
    \begin{align*}
        g(t)=\tilde{h}^{*}(1)-\tilde{h}^{*}(1-t)=n\tilde{h}\left(\frac{1}{n}\right)-n\tilde{h}\left(\frac{1-t}{n}\right)=nh\left(1-\frac{1-t}{n}\right)-nh\left(1-\frac{1}{n}\right).
    \end{align*}
    Similar as (i), 
    we can always find a partition $(A_{1}, \ldots, A_{n})\in \Pi_{n}$ independent of $X$ such that $\mathbb{P}(A_{i})=1/n$ for all $i\in[n]$.
    Then we can obtain that 
     \begin{align*}
        \sum_{i=1}^{n}\rho_{h}(X\mathbb{1}_{A_{i}})= -\int_{-\infty}^{0} n\tilde{h}\left(\frac{1}{n}\mathbb{P}(X \leq x)\right) \mathrm{d} x = \rho_{g}(X), 
    \end{align*}
    which implies that $I(X)\leq \rho_{g}(X)$.
    Therefore, the desired results are obtained.
\end{proof}
The next result  derives explicit formulas of the unconstrained and counter-monotonic inf-convolution of $n$ concave distortion risk measures, as well as the optimal allocation. These formulas vary depending on the set of bounded random variables considered.
In particular, when total risk $X$ are nonnegative or nonpositive and we restrict the set of allocations to allocations satisfying $X_{i}\geq 0$ or $X_{i}\leq 0$,
both the constrained inf-convolution and the counter-monotonic inf-convolution are identified as risk metrics. Furthermore, when both $X$ and $X_i$ are general bounded random variables, the counter-monotonic inf-convolution is determined to be negative infinity.
The above results are mainly proved by using the technique of counter-monotonic theorem and Lemma \ref{lemma:convex}.

\begin{theorem}\label{theorem:n_convex}
Let $h\in \mathcal{H}$ be  convex. The following hold. 
\begin{itemize}
    \item[(i)] If $\mathcal X=L^+$ and $X\in \mathcal{X}^{\perp}$, then $$\dsquare_{i=1}^n \rho_{h}(X)= \dboxminus_{i=1}^{n}\rho_{h}(X)=  \rho_{g}(X),$$
    where $g(t)=n h(t/n)$ for $t\in [0,1]$. 
    
    \item[(ii)] If $\mathcal X=L^-$ and  $X\in \mathcal{X}^{\perp}$, then 
    $$\dsquare_{i=1}^n \rho_{h}(X)= \dboxminus_{i=1}^{n}\rho_{h}(X)= \rho_{g}(X),$$
    where $g(t)=nh(1-(1-t)/n)-nh(1-1/n)$ for $t\in [0,1]$.  
    \item[(iii)] 
    If $\mathcal X=L^\infty$,  $X\in \mathcal X^\perp$ and $h $ is  not the identity, then $$\dsquare_{i=1}^n \rho_{h}(X)= \dboxminus_{i=1}^{n}\rho_{h}(X)= -\infty.$$
\end{itemize} 
    Moreover,  in (i) and (ii), any uniform counter-monotonic allocation 
    is Pareto-optimal. 
\end{theorem}

\begin{proof}
    (i) Suppose $\mathcal X=L^+$ and $X\in \mathcal X^{\perp}$. Let $(X_{1}, \ldots, X_{n})\in \mathbb{A}_{n}(X) $.
    By counter-monotonic improvement theorem, there always exists a jackpot allocation $(Y_{1}, \ldots, Y_{n} )\in \mathbb{A}_{n}(X)$ such that $Y_{i}=X\mathbb{1}_{A_{i}}$, 
    $Y_{i}\leq_{\mathrm{cv}} X_i$ and $Y_{i}\geq 0$ for each $i\in [n]$ and $(A_{1}, \ldots, A_{n})\in \Pi_{n}$. If $h\in \mathcal{H}$ is convex, then it holds $\rho_{h}(X\mathbb{1}_{A_{i}})\leq \rho_{h}(X_{i})$ for all $i\in[n]$.
    Denote by $I(X)=\inf \left\{\sum_{i=1}^{n}\rho_{h}(X\mathbb{1}_{A_{i}}), (A_{1}, \ldots, A_{n})\in \Pi_{n}\right\}$. 
    It follows that 
    \begin{align}\label{ineq:convex}
        \dboxminus_{i=1}^{n} \rho_{h}(X) \leq I(X) \leq \dsquare_{i=1}^{n} \rho_{h}(X).
    \end{align}
    Also, it is straightforward to verify $\dsquare_{i=1}^{n} \rho_{h}(X) \leq \dboxminus_{i=1}^{n} \rho_{h}(X)$. From Lemma \ref{lemma:convex}, $I(X)$ is determined by $ \rho_{ h^{*}}(X)$, where $h^{*}(t)=n h(t/n)$, $t\in[0,1]$.

    (ii) 
    Next we assume $\mathcal X=L^-$ and $X\in \mathcal X^{\perp}$. 
    By counter-monotonic improvement theorem, for any $(X_{1}, \ldots, X_{n})\in \mathbb{A}_{n}(X)$, there exists a jackpot allocation $(Y_{1}, \ldots, Y_{n})\in \mathbb{A}_{n}(X)$ such that $Y_{i}\leq_{\mathrm{cv}} X_i$ and $Y_{i}\leq 0$ for each $i\in [n]$.
    Similarly with the analysis in (i), the inequality (\ref{ineq:convex}) also holds true in the case of $\mathcal X=L^-$ but with different $I(X)$. From (ii) of Lemma \ref{lemma:convex}, we can obtain 
    \begin{align}\label{ineq:neg}
        \dboxminus_{i=1}^{n} \rho_{h}(X) =\dsquare_{i=1}^{n} \rho_{h}(X)= \rho_{g}(X).
    \end{align}
    where $g(t)=nh(1-(1-t)/n)-nh(1-1/n)$. 

    (iii) In this case, we consider $\mathcal{X}=L^\infty$ and  $X\in \mathcal{X}^{\perp}$. Take $m\in \mathbb{R}_{+}$ large enough such that $X+m\ge 0$. By translation invariance of $\rho_h$, it holds that
    \begin{align*}
        \dboxminus_{i=1}^{n} \rho_{h}(X)
        &= \inf \left\{\sum_{i=1}^{n}\rho_{h}\left(X_{i}+\frac{m}{n}\right): (X_{1}, \ldots, X_{n})\in \mathbb{A}_{n}^{-}(X)\right\}-m\\
        &=  \inf \left\{\sum_{i=1}^{n}\rho_{h}(Y_{i}): (Y_{1}, \ldots, Y_{n})\in \mathbb{A}_{n}^{-}(X+m)\right\}-m \\
        &\le \inf \left\{\sum_{i=1}^{n}\rho_{h}(Y_{i}): (Y_{1}, \ldots, Y_{n})\in \mathbb{A}_{n}^{-}(X+m),~Y_1,\dots,Y_n\ge 0\right\}-m 
     \\&= \rho_g(X+m)-m ,
    \end{align*}
where $g$ is given by $g(t)=nh(t/n)$ for $t\in [0,1]$, and the last equality is due to part (i).
 Since $h$ is convex and not the identity, we have $g(1)<1$, and hence $ 
    \rho_g(X+m)-m  = \rho_g(X)+mg(1)-m \to -\infty$ as $m\to\infty$.
  Therefore,  $\dboxminus_{i=1}^{n} \rho_{h}(X) 
  =-\infty$.
\end{proof}

    When the total risk $X$ is either nonnegative or nonpositive, we observe that  $\rho_g$  in both  (i) and (ii) of Theorem \ref{theorem:n_convex} is no longer a distortion risk measure as $g(1) \neq 1$, but it generally belongs to the class of distortion riskmetrics. 
    Given a convex $h$, 
    it directly follows from Theorem \ref{theorem:n_convex} that $\rho_g(X) \leq \rho_h(X)$ since $\dsquare_{i=1}^{n}\rho_{h}(X)=\dboxminus_{i=1}^{n}\rho_{h}(X)\leq \rho_{h}(X)$. 
     Notably, 
     in the case of $X$ being nonpositive, $\rho_g(X) \leq \rho_h(X)$ holds even though $g \geq h$; recall that if $f,h\in \mathcal H$, then 
    Lemma 1 of \cite{wang2020characterization} implies that $\rho_{f} \le \rho_{h}$ if and only if $f\le h$, and this emphasizes  $g\not \in \mathcal H$. 
     We provide a numerical example to show the relation of $\dsquare_{i=1}^n \rho_{h} \leq \rho_h$ with $h$ being convex. In this example, we consider a scenario with two agents in the pool. 
     Take $h(t)=\Phi(\Phi^{-1}(t)+\lambda)$ and $\lambda=-0.6$. Some numerical results are presented in Table \ref{table:3}.
\begin{table}[ht!]
\renewcommand{\arraystretch}{1.5}
\centering
\begin{tabular}{c|c|c|c|cc} 
 & $X$ & $\X$ &$\rho_h (X)=  \dboxplus_{i=1}^2 \rho_{h}(X)$ &  $\dboxminus_{i=1}^2 \rho_{h}(X)= \dsquare_{i=1}^2 \rho_{h}(X)$\\ 
\hline
\multirow{2}{*}{$Y \sim \text{Uniform}(0,1)$ } &$Y$ & $L^+$  & 0.3317 & 0.1903  \\
& $-Y$ & $L^-$ & -0.6609 & -0.8776  \\
\hline
\multirow{2}{*}{$Y \sim \text{Pareto}(3,2)$} &$Y$ & $L^+$   & 2.4743 & 1.4406  \\
& $-Y$  & $L^-$ & -3.6044 & -4.9292 \\
\hline
\multirow{2}{*}{$Y \sim \text{logN}(0,1)$} &$Y$ & $L^+$  & 0.92704 & 0.6408  \\
& $-Y$  & $L^-$ & -3.0062 & -3.5515 \\
\hline
\end{tabular}
\label{table:3}
\caption{Comparison of the three inf-convolutions.}
\end{table}

\begin{example}
We provide an example to illustrate how inf-convolutions vary with different risk preferences by specifying the distortion function as $h(t)=1-(1-t)^{\alpha}, \alpha\in \mathbb{R}_{+}$. It is clear that agents exhibit risk-seeking (RS) behavior with $\alpha < 1$ due to the convexity of $h$, whereas they are risk-averse (RA) with $\alpha > 1$ and risk-neutral (RN) when $\alpha =1$, indicating a linear $h$. 
It is well-known that $\dboxplus_{i=1}^{n}\rho_{h}$ is always equal to $\rho_{h}$, regardless of whether $h$ is concave or convex, as shown in the blue line in Figure \ref{fig:inf_convolution}.
From Theorem \ref{theorem:counter_two}, these three inf-convolutions are equal when $h$ is concave, i.e., $\alpha >1$ in this case.
Consequently, they share the same blue line in scenarios where agents are risk-averse. 
Furthermore, Theorem \ref{theorem:n_convex} shows 
that $\rho_{h}$ is consistently greater than both the counter-monotonic inf-convolution and the unconstrained one. 
This leads to the depicted deviation, represented by the red line in Figure \ref{fig:inf_convolution}, for risk-seeking agents. 


\begin{figure}[ht]
  \centering
\begin{tikzpicture}
\begin{axis}[
    axis lines=middle,
    enlargelimits=true,
    clip=false,
    axis line style={-stealth},
    no marks,
    xmin = 0.45, ymin =0.14,
    xtick=\empty,  
    ytick=\empty,  
    xlabel = $\alpha$,
    xlabel style={at={(axis description cs:1,0)}, anchor=north},
    ylabel = \empty,
    axis line style={line width=0.8pt}, 
]
\fill [blue!10] (axis cs:0.2,0.06) rectangle (axis cs:1,1);
\fill [black!10] (axis cs:1,0.06) rectangle (axis cs:3,1);
\addplot+[thick, color=blue, opacity=0.7, domain=0.2:1, samples=100] {1 - (1 - 3/4)^x};
\addplot+[thick, color=blue, opacity=0.7, domain=1:3, samples=100] {1 - (1 - 3/4)^x};
\node at (axis cs:2.2,0.8) [anchor=north, black] {$ \dboxminus_{i=1}^{n}\rho_{h_A}=\dsquare_{i=1}^{n}\rho_{h_A}=\dboxplus_{i=1}^{n}\rho_{h_A}$};

\node at (axis cs:0.5,0.75) [anchor=north, black] {$\dboxplus_{i=1}^{n}\rho_{h_S}$};

\node at (axis cs:1.2,0.35) [anchor=north, black] {$\dboxminus_{i=1}^{n}\rho_{h_S}=\dsquare_{i=1}^{n}\rho_{h_S}$};
\addplot+[thick, color=red, opacity=0.6, domain=0.2:1, samples=1000] {20*(1 - (1-3/4*1/20)^x)};
\draw [dashed, thick] (axis cs:1,0.06) -- (axis cs:1,0.99);

\node at (axis cs:1.2,0.05) [anchor=north] {1 (RN)};

\node at (axis cs:0.5,0.05) [anchor=north, black] {RS};

\node at (axis cs:2,0.05) [anchor=north, black] {RA};
\end{axis}
\end{tikzpicture}
  \caption{Comparison of inf-convolutions evaluated at $X\sim \mathrm{Uniform}[0,1]$ for agents with different risk attitudes.}
  \label{fig:inf_convolution}
\end{figure}
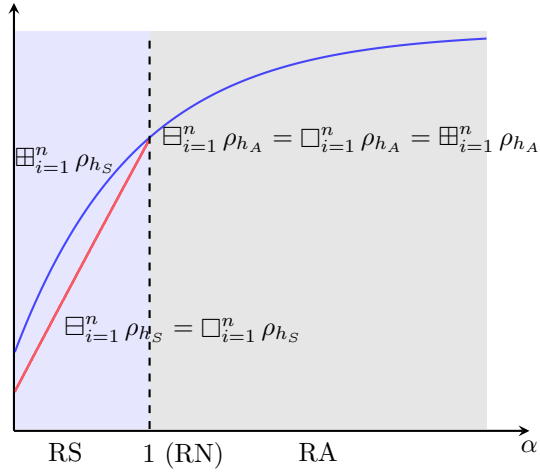
\end{example}

It is natural to wonder about the connections among comonotonicity, counter-monotonicity and Pareto optimality. Next, we characterize Pareto-optimal allocations of the risk sharing problem with agents exhibiting risk-averse or risk-seeking behaviors.
The following proposition shows that, within the framework of distortion risk measures, every comonotonic allocation is Pareto optimal for all risk-averse agents, and the converse holds true for strictly concave distortion functions. In contrast, for risk-seeking agents, Pareto-optimal allocations must be jackpot (scapegoat) allocations if the total risk $X$ is nonnegative (nonpositive), complementing the result in Theorem \ref{theorem:n_convex}. 

\begin{proposition}
\label{prop:summary}
    Assume $h\in \mathcal{H}$. 
    The following statements hold.
    \begin{itemize}
        \item[(i)] Suppose $\mathcal{X} = L^\infty$ and $X\in \mathcal{X}$. 
        If $h$ is concave, then all comonotonic allocations of $X$ are Pareto optimal.
        If $h$ is strictly concave, then 
        all 
         Pareto-optimal allocations of $X$ are comonotonic.
        \item[(ii)]Suppose $\mathcal{X} = L^{+}$ (resp.~$\X=L^{-}$) and $X\in \mathcal{X}^{\perp}$. If $h$ is convex, then all Pareto-optimal allocations are jackpot (resp.~scapegoat) allocations.
    \end{itemize}
\end{proposition}
\begin{proof}
     (i) The first statement follows from Theorem \ref{th:three_h} and comonotonic additivity of $\rho_h$.  The second statement follows from Proposition 4 of \cite{lauzier2023risk}. 
     (ii) When $\mathcal{X}=L^{+}$, the result directly follows Theorem 2 of \cite{lauzier2024negatively}. The proof for nonpositive case is analogous and is thus omitted.
\end{proof}

\section{Portfolio manager's problem}\label{sec:port}

In this section, we analyze a portfolio optimization problem, showing how the risk preferences of agents influence their decision-making regarding risky investments. In the market, there is a portfolio manager and a group of agents who collectively possess a constant initial endowment of  $W$ and homogeneous individual risk preferences, which can vary from highly risk-averse to potentially risk-seeking. The task of a portfolio manager is to manage the investments for a group of clients, aiming to optimize the collective portfolio in a way that aligns with the risk preference of the agents and to minimize the aggregate risks. The investment payoff from a risky asset, such as a stock or mutual fund, after the investment period at time $1$ is modelled by a random variable $X\ge 0$. The first question for the manager is to construct an investment strategy $\lambda X$, where $\lambda$ represents the total proportion of the total investment to allocate to the risky asset.
This investment has a cost   $c(\lambda)$ incurred from investing in the risky asset and $c$ is   assumed to be increasing and convex, meaning that a larger leverage is marginally more costly (see examples in \cite{follmer2002convex} and \cite{castagnoli2022star}). 
The second task is to allocate the total investment  wealth $W+\lambda X-c(\lambda)$ at time $1$ to the participants via an allocation $(X_1,\dots,X_n)$ of wealth.

To summarize, the goal of the manager, taking into account both the investment problem and the allocation problem, is to optimize the following objective function
\begin{align}
\begin{aligned} 
\text{to minimize}~~ &  \sum_{i=1}^n \rho_{h}(-X_i)  
 \\  
\text{subject to}~~& \lambda \in [0,1],~~
 c(\lambda)\leq W; \\&  X_1+\dots + X_n=
 W+\lambda X-c(\lambda); ~~X_1,\dots,X_n\ge 0.
 \end{aligned} \label{problem:por} 
\end{align}

In this model, the constraint $X_1,\dots,X_n\ge 0$ means that the manager does not give additional loss to the participants at time $1$. In other words, the participation fund represented by $W$ has been collected in the beginning of the investment period.
To avoid infeasibility, we assume  $c(\lambda)\leq W$ to ensure the total wealth $W+\lambda X-c(\lambda)$ remains nonnegative. 
We will determine the optimal proportion $\lambda$ of risky investments in the portfolio, which is dependent on the risk attitudes of agents, as shown in the subsequent result.
We aim to investigate how these risk preferences influence the allocation of risky investments. Specifically, we consider scenarios where agents are either risk-averse or risk-seeking.

\begin{proposition}\label{prop:lambda}
    Suppose $\mathcal{X}=L^+$ and $h\in \mathcal{H}$.
    For $X\in \mathcal{X}^{\perp}$, the following hold.
    \begin{itemize}
        \item[(i)] If $h$ is concave, the optimal value for problem \eqref{problem:por} is 
        \begin{align*}
            \lambda^{*} = \min \left\{
    {c^{\prime}}^{-1}(\rho_{\tilde{h}}(X)), {c^{\prime}}^{-1}(W)
    \right\}, \quad  \text{where} \  \tilde{h}(t)=1-h(1-t).
        \end{align*}
         \item[(ii)] If $h$ is convex, the optimal value is 
         \begin{align*}
         \lambda^{*} = \min \left\{{c^{\prime}}^{-1}(\rho_{g}(X)), {c^{\prime}}^{-1}(W)\right\}, \quad \text{where} \ g(t)=\frac{1-h(1-t/n)}{1-h(1-1/n)}.
        \end{align*}
    \end{itemize}
\end{proposition}
\begin{proof}
    (i) Denote by $f(\lambda)=\rho_{h}(-W-\lambda X+c(\lambda))$. 
    From concavity of $h$ and Theorem \ref{th:h_n}, problem \eqref{problem:por} becomes finding the minimum of $f(\lambda)$ over $\lambda\in [0,1)$ with constraint of $c(\lambda)\leq W$. Furthermore, the function $f(\lambda)$ can be rewritten as
    \begin{align*}
        f(\lambda)= -\lambda\rho_{\tilde{h}}(X)+ c(\lambda)-W.
    \end{align*}
    Take derivative with respect to $\lambda$, we have 
    $f^{\prime}(\lambda) = -\rho_{\tilde{h}}(X)+ c^{\prime}(\lambda)=0$. 
    Consider the restriction of $c(\lambda)\leq W$, the optimal value $\lambda^*$ is then determined by $\lambda^{*} = \min  \{
    {c^{\prime}}^{-1}(\rho_{\tilde{h}}(X)), {c^{\prime}}^{-1}(W)
    \}$, where $\tilde{h}(t)=1-h(1-t)$.

    (ii) The result can be proved using Theorem \ref{theorem:n_convex} and convexity of $h$. Its proof is similar to (i) and thus omitted.
\end{proof}

An immediate implication of the above results is that for risk-averse agents, the optimal proportion of risky investment remains unaffected by the number of agents involved in the pool, whereas this is not the case for risk-seeking agents. In the following proposition, we present a necessary and sufficient condition for the optimal strategy to be independent with $n$ for risk-seeking agents. 
The assumptions in Proposition \ref{prop:lambda} are maintained in the following results. 
\begin{proposition}\label{prop:power}
    If $h$ is convex, the optimal value $\lambda^*$ is independent of $n$ if and only if $h$ is the dual-power transform,  $h(t)=1-(1-t)^{\alpha}$ with $\alpha \in (0,1]$.
\end{proposition}
\begin{proof}
    For the ``if" part, it is trivial to show the optimal value $\lambda^*$ is unrelated to $n$ by (ii) of Proposition \ref{prop:lambda}. Furthermore, the optimal value $\lambda^*$ is determined as 
    \begin{align*}
        \lambda^{*} = {c^{\prime}}^{-1}\left(\rho_{g}(X)\right), ~~~ \text{where}~ g(t)=\tilde{h}(t)=x^\alpha.
    \end{align*}

    Next, we show the ``only if" part. From Proposition \ref{prop:lambda}, independence between  the optimal value $\lambda^*$ and $n$ for $h$ being convex is equivalent to 
    independence between 
    $g(t)=(1-h(1-t/n))/(1-h(1-1/n))$ and $n$. The function $g(x)$ can be rewritten as $g(t)= \tilde{h}(t/n)/\tilde{h}(1/n)$. It follows that the equality 
    \begin{align*}
        \frac{\tilde{h}\left(\frac{t}{n_1}\right)}{\tilde{h}\left(\frac{1}{n_1}\right)}=\frac{\tilde{h}\left(\frac{t}{n_2}\right)}{\tilde{h}\left(\frac{1}{n_2}\right)}
    \end{align*}
    holds for any $n_1 \neq n_2  \in \mathbb{Z}$. Take $n_1 = 1$ and $n_2 = n$. It follows   that 
    \begin{align}\label{eq:1}
        \tilde{h}\left(\frac{t}{n}\right) = \tilde{h}\left(\frac{1}{n}\right) \tilde{h}(t).
    \end{align}
    Taking the derivative with respect to $x$ on both sides, we get 
    \begin{align}\label{eq:2}
        \frac{1}{n}\tilde{h}^{\prime}\left(\frac{t}{n}\right) = \tilde{h}\left(\frac{1}{n}\right) \tilde{h}^{\prime}(t).
    \end{align}
    Dividing \eqref{eq:1} by \eqref{eq:2}, we get 
    \begin{align*}
    \ell (t)=n  \ell \left(\frac{t}{n}\right), ~~~ \text{where} \  \ell (t)=\frac{\tilde{h}(t)}{\tilde{h}^{\prime}(t)}.
    \end{align*}
    Hence,  the above property implies that $
    \ell$  has the form of $ \ell (t)=\beta t$ for some $\beta\ge 0$ and all $t\in (0,1]$. As a consequence,
    \begin{align*}
        \frac{\tilde{h}^{\prime}(t)}{\tilde{h}(t)} = \frac{\alpha}{t}, \, ~~\text{where} \ \alpha =\frac{1}{\beta}
    \end{align*}
    and thus $\log \tilde{h}(t)=\alpha \log t$. Therefore, the distortion function $h(t)$ has the form of $h(t)=1-(1-t)^{\alpha}$ where $\alpha \in (0,1]$. 
\end{proof}

We give two simple examples to see how the optimal proportion of risky investments varies with the number of agents and the risk attitudes of agents. Here the cost function is specified as $c(\lambda)=\lambda^2/2$.
\begin{example}\label{ex:power}
    Consider agents are associated with $h(t)=1-(1-t)^{\alpha}$. It is clear that $h$ is concave with $\alpha \in [1, \infty)$ and is convex with $\alpha \in (0,1]$. 
    The investment $X$ here is assumed to follow  the uniform distribution $\mathrm U(0,1)$. We aim to explore the relationship between the optimal value $\lambda^{*}$ and the risk preference characterized by $\alpha$. Some numerical results are presented in Figure \ref{fig:power}.
    \begin{figure}[ht]
    \centering
    \includegraphics[width=8cm]{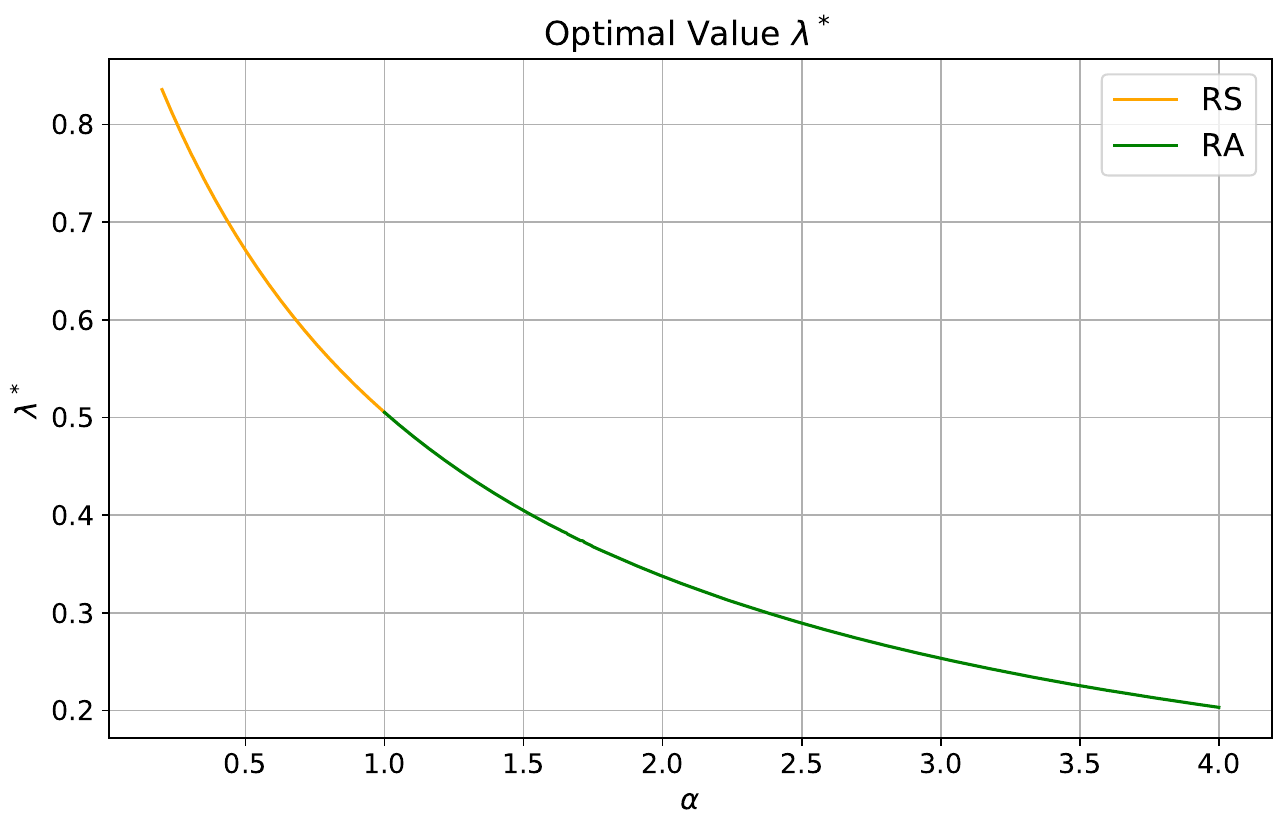}
    \caption{Optimal value of $\lambda^*$ with $h(t)=1-(1-t)^{\alpha}.$}
    \label{fig:power}
    \end{figure} 
    Proposition \ref{prop:power} states the optimal value $\lambda^{*}$ would not vary with $n$, and  that is the reason why we only have one single line for $\alpha \in (0,1]$. From Figure \ref{fig:power}, we observe that the optimal value $\lambda^*$ decreases as $\alpha$ increases, implying that the manager on behalf of risk-averse agents tends to invest less in risky assets compared to the case of risk-seeking agents.  
\end{example}

\begin{example}\label{ex:wang}
    Consider a distortion function, $h(t)=\Phi(\Phi^{-1}(t)+\alpha)$, known as a   transform of \cite{wang2000class}. Agents are risk-averse if $\alpha \in [0, \infty)$, while they are risk-seeking if $\alpha \in (-\infty, 0]$. Take a uniform random variable $X \sim \mathrm U(0,1)$. The optimal value of $\lambda$ is computed by varying the number of agents within the group and adjusting the risk preference parameter $\alpha$ from negative to positive values. Numerical results are presented in the Figure \ref{fig:wang}.
    \begin{figure}[ht]
    \centering
    \includegraphics[width=8.5cm]{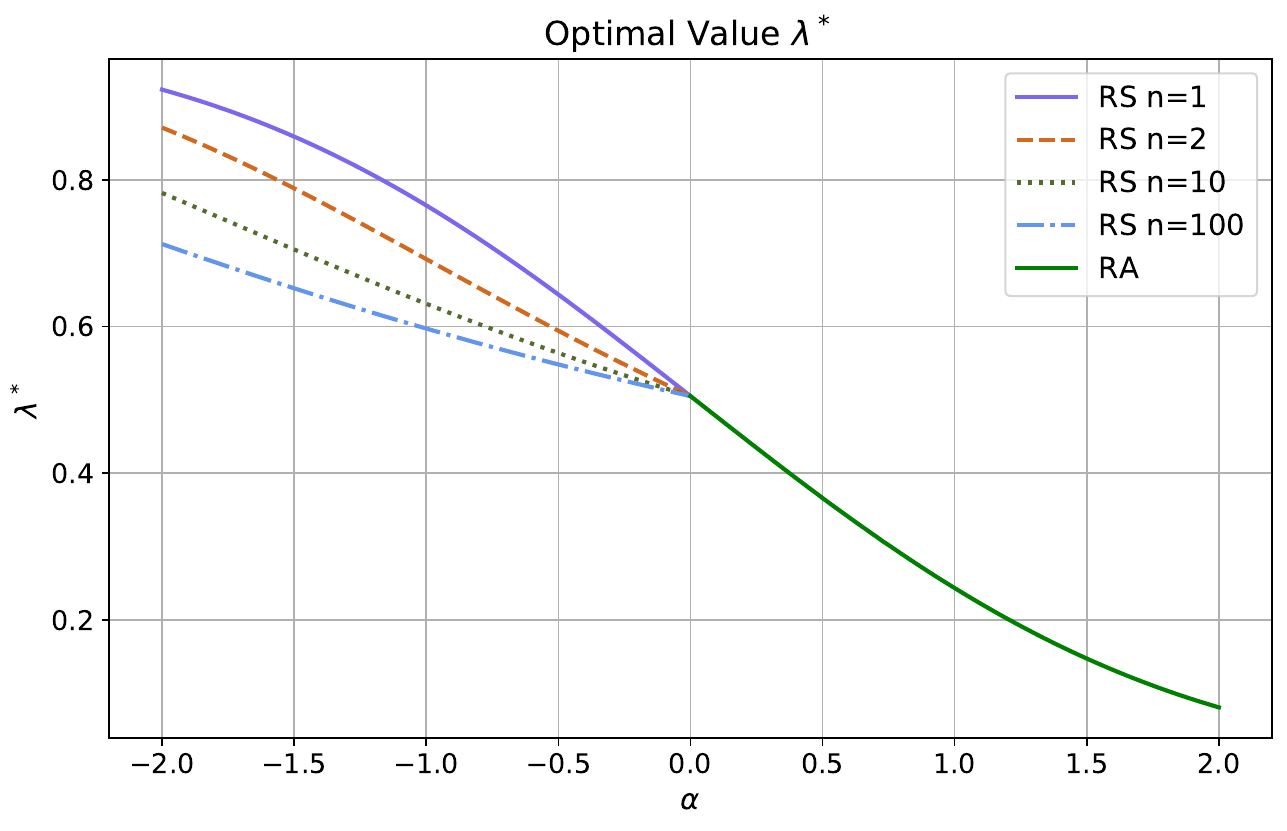}
    \caption{Optimal value of $\lambda^*$ with $h(t)=\Phi(\Phi^{-1}(t)+\alpha).$}
    \label{fig:wang}
    \end{figure} 
    Similar to Example \ref{ex:power}, the more risk-seeking the agents are, the greater the allocation of risky investments. In this case, the optimal value   $\lambda^*$ varies with different sizes of participants and decreases as the number of agents   increases. 
    The economic intuition is that as more risk-seeking agents join the group, they tend to engage more in gambling with each other. Consequently, the randomness in outcomes arises from two sources: the gambling itself and the risky investment. With an increase in risk-seeking agents in the group, there is a tendency to focus on gambling among themselves rather than increasing investment in  the financial market due to its increasing marginal cost; recall that gambling among participants themselves incur no additional cost for the group. It is natural to wonder if this is always the case; that is, whether the optimal proportion of risky investments will always decrease with an increasing number of agents in the pool. The answer is negative. We provide a counter-example in Appendix \ref{sec:port_exp} that shows the optimal proportion of risky investments can actually increase with more agents involved.
\end{example}

\section{Inverse S-shaped distortion functions}\label{sec:inverse}
We now consider another risk sharing  problem among agents with a special distortion risk measure which is neither concave nor convex.  \cite{tversky1992advances} introduced an inverse S-shaped distorted probability function expressed as $$h_{\mathrm{KT}}(t)=\frac{t^\gamma}{\left(t^\gamma+(1-t)^\gamma\right)^{1 / \gamma}},$$ where $\gamma \in (0,1)$. This weighting function, also known as a KT distortion function, exhibits a concave-convex shape, indicating a mix of risk-seeking and risk-averse preferences. 

\begin{assumption}\label{assum:1}
    The distortion risk functional $h:[0,1] \mapsto [0,1]$ is concave-convex.
\end{assumption}

Denoting by $\underline{h}:[0,1] \rightarrow[0,1]$ the convex envelope of $h$, which is defined as the largest convex function such that $\underline{h}(t) \leq h(t)$ for all $t \in[0,1]$. Clearly, $\underline{h}$ is a distorted probability function being dominated pointwise by $h$. Thus, following from the Lebesgue–Stieltjes integral representation of $\rho_h(X)$ (see \cite{dhaene2012remarks}), for any $X \in \mathcal{X}$ we have 
$$
\rho_{\underline{h}}(X)=\int_0^1 F_X^{-1}(1-t) \mathrm{d} \underline{h}(t) \leq \int_0^1 F_X^{-1}(1-t) \mathrm{d} h(t)=\rho_h(X).
$$

Next, we explore the risk sharing problem for agents whose risk preferences are modeled by concave-convex distortion functions. The subsequent theorem presents an explicit formula for the corresponding counter-monotonic inf-convolution, which depends  on the convex envelope of the distortion functions. This result has a condition that holds if the number of agents is involved in the risk sharing problem is larger than a constant $1/(1-t_0)$, discussed later.

\begin{theorem}\label{theorem:concave_convex}
    Suppose $\mathcal X=L^-$ and  $X\in \mathcal{X}^{\perp}$. Assume $h\in \mathcal{H}$ satisfies Assumption \ref{assum:1}. 
    If $n\geq 1/(1-t_{0})$, where $t_{0}=\sup \left\{t\in (0,1]: h^{\prime}(t)< h(t)/t\right\}$,
    then it holds
    \begin{align*}
        \dboxminus_{i=1}^{n} \rho_{h}(X) = \dboxminus_{i=1}^{n} \rho_{\underline{h}}(X) = \rho_{g}(X),
    \end{align*}
    where $g(t)=n \underline{h}(1-(1-t)/n)-n \underline{h}(1-1/n)$ for $t\in [0,1]$.  
\end{theorem}
\begin{proof}
    Assume $h\in \mathcal{H}$ is concave-convex with  $t_{0}=\sup \left\{t\in (0,1]: h^{\prime}(t)< h(t)/t\right\}$.
    One can check that
    $h$ is convex on $[t_0,1]$ and identical to $\underline h$;  see e.g., 
    \citet[Lemma A.8]{ghossoub2019optimal}.
    Thus, the convex envelope $\underline{h}(t)$ is determined as 
    \begin{equation}
    \underline{h}(t)= \begin{cases}\frac{h(t_{0})}{t_{0}}t & \text { if } t \in\left[0, t_0\right] \\ h(t) & \text { if } t \in\left[t_0, 1\right]\end{cases}.
    \end{equation}
    If $n\geq 1/(1-t_{0})$, it is trivial to verify $h(t)=\underline{h}(t)$ for $t\in [1-1/n, 1]$.
    
    Let $\mathcal X=L^-$. 
    For $X\in \mathcal{X}^{\perp}$, it is trivial to verify $\dboxminus_{i=1}^{n} \rho_{\underline{h}}(X) \leq \dboxminus_{i=1}^{n} \rho_{h}(X)$, which
    directly follows from $\underline{h}$ being dominated by $h$ pointwisely. Conversely, let $(X_{1}, \ldots, X_{n})\in \mathbb{A}_{n}^{-}(X)$ be a jackpot allocation, i.e., $X_{i}=X\mathbb{1}_{A_{i}}$ with $\mathbb{P}(A_{i})=1/n$ for all $i\in [n]$ and $(A_{1}, \ldots, A_{n})\in \Pi_{n}$ independent of $X$. For such $(A_{1}, \ldots, A_{n})\in \Pi_{n}$ and $X\in \mathcal{X}^{\perp}$, it follows that
    \begin{align*}
        \sum_{i=1}^{n}\rho_{h}(X\mathbb{1}_{A_{i}}) 
        &= n\int_{-\infty}^{0}
        \left(h\left(\frac{1}{n}\mathbb{P}(X>x)+1-\frac{1}{n}\right)-1 \right)\,\mathrm{d}x \\
        &= n\int_{-\infty}^{0}
        \left(\underline{h}\left(1-\frac{1}{n}\mathbb{P}(X\leq x)\right)-1 \right)\,\mathrm{d}x\\
        &= n\int_{0}^{\infty}
        \left(\underline{h}\left(1-\frac{1}{n}\mathbb{P}(-X\geq x)\right)-1 \right)\,\mathrm{d}x \\
        &= - n \int_{0}^{\infty} \tilde{\underline{h}}\left(\frac{1}{n}\mathbb{P}(-X\geq x)\right) \,\mathrm{d}x \\
        & = -\rho_{\tilde{\underline{h}}^{*}}(-X)=\rho_{g}(X),
    \end{align*}
    where $\tilde{\underline{h}}(t)=1-\underline{h}(1-t)$, $\tilde{\underline{h}}^{*}(t)=n\tilde{\underline{h}}(t/n)$ and $g(t)=\tilde{\underline{h}}^{*}(1)-\tilde{\underline{h}}^{*}(1-t) $, which can be further simplified as $g(t)=n\underline{h}(1-(1-t)/n)-n\underline{h}(1-1/n)$. 
    Furthermore, we can verify $\rho_{g}(X)=\dboxminus_{i=1}^{n} \rho_{\underline{h}}(X)$ from convexity of $\underline{h}$ and (ii) of Theorem \ref{theorem:n_convex}. Thus,    
    $\dboxminus_{i=1}^{n} \rho_{h}(X) \leq \dboxminus_{i=1}^{n} \rho_{\underline{h}}(X)$ for $X\in \mathcal{X}^{\perp}$.
    Combining the above arguments, the desired result follows.
\end{proof}

The typically value $t_0$ in  Theorem \ref{theorem:concave_convex}
 is not close to $1$ in behaviour economics,
 and hence $1/(1-t_0)$ is not very large.
 For instance, 
in the classic work of \cite{wu1996curvature}, 
the KT  distortion function has a best estimated parameter $\gamma=0.71$, which corresponds to $t_0\approx 0.768$; see Figure 6 of \cite{wu1996curvature}.
Therefore, it suffices for $n\ge 5$ to apply the result of Theorem \ref{theorem:concave_convex} in that model.

An immediate consequence of Theorem \ref{theorem:concave_convex} is that when a good number of agents are involved in the risk pool and their preferences are modeled by inverse S-shaped distortion risk measures, these agents will behave like risk-seeking agents discussed in Section \ref{sec:convex}; that is, they will achieve the same optimal value and share the same optimal allocations. The intuitive economic explanation is that when the number of agents in the pool is not too small, agents may prefer gambling against each other to improve the outcome of the risk minimization problem
since those agents in Theorem \ref{theorem:concave_convex} are non-risk-averse.

Theorem \ref{theorem:concave_convex} implies that, with the space of risks being nonpositive random variables,   scapegoat allocations are Pareto optimal for agents with inverse S-shaped distortion functions under some mild conditions. 
As we have seen from Proposition \ref{prop:summary}, this is the same situation for risk-seeking agents.

\section{Conclusion}
\label{sec:conclusion}

The comonotonic risk sharing problem has been well studied, and the comonotonic inf-convolution of $n$ distortion risk measures can be determined explicitly, as studied extensitve in the literature. It is well-known that comonotonicity being optimal is a consequence of the concavity of distortion functions. 
Our paper addresses different situations from the classic literature, where the agents are possibly associated with convex distortion functions. 

Our study mainly focuses on counter-monotonic risk sharing problems for agents with homogeneous risk distortion functions and provides a comparative analysis of three types of risk sharing problems: unconstrained, comonotonic, and counter-monotonic. This analysis is conducted for cases where the risk is pooled among risk-averse agents, risk-seeking agents, and agents with inverse S-shaped distortion functions.
We provide explicit formulas of counter-monotonic inf-convolution for each of these three scenarios. In addition, our results  provide insights into solving the unconstrained risk sharing problem for some non-concave distortion functions, i.e., convex or inverse S-shaped functions, typically leading to counter-monotonic Pareto-optimal allocations (under some mild conditions for the latter one).

By combining the results for risk-averse agents and risk-seeking agents, we are able to solve the portfolio optimization problem as described in Section \ref{sec:port}, and the optimal strategies for risky assets can be determined explicitly. Some examples are presented in the Section \ref{sec:port} and Appendix 
 \ref{sec:port_exp}.

The explicit results obtained in this paper assumed that the preferences are the agents are homogeneous, that is, using the same risk measure. 
The case of heterogeneous risk measures 
imposes substantial technical challenges and will need separate analysis. 

 \bibliography{ref}

\begin{thebibliography}{}

\bibitem[Aouani et~al., 2021]{aouani2021propensity}
Aouani, Z., Chateauneuf, A., and Ventura, C. (2021).
\newblock Propensity for hedging and ambiguity aversion.
\newblock {\em Journal of Mathematical Economics}, 97:102543.

\bibitem[Araujo et~al., 2018]{ACGNecta2018}
Araujo, A., Chateauneuf, A., Gama, J.~P., and Novinski, R. (2018).
\newblock {General equilibrium with uncertainty loving preferences}.
\newblock {\em Econometrica}, 86(5):1859--1871.

\bibitem[Araujo et~al., 2022]{AGS2022}
Araujo, A., Gama, J., and Suarez, C. (2022).
\newblock {Lack of prevalence of the endowment effect: An equilibrium analysis}.
\newblock {\em Journal of Mathematical Economics}, 102:102763.

\bibitem[Artzner et~al., 1999]{artzner1999coherent}
Artzner, P., Delbaen, F., Eber, J.-M., and Heath, D. (1999).
\newblock Coherent measures of risk.
\newblock {\em Mathematical Finance}, 9(3):203--228.

\bibitem[Barrieu and El~Karoui, 2005]{barrieu2005inf}
Barrieu, P. and El~Karoui, N. (2005).
\newblock Inf-convolution of risk measures and optimal risk transfer.
\newblock {\em Finance and Stochastics}, 9(2):269--298.

\bibitem[Bei{\ss}ner and Werner, 2023]{beissner2023optimal}
Bei{\ss}ner, P. and Werner, J. (2023).
\newblock Optimal allocations with $\alpha$-maxmin utilities, choquet expected utilities, and prospect theory.
\newblock {\em Theoretical Economics}, 18(3):993--1022.

\bibitem[Boonen et~al., 2021]{boonen2021competitive}
Boonen, T.~J., Liu, F., and Wang, R. (2021).
\newblock Competitive equilibria in a comonotone market.
\newblock {\em Economic Theory}, 72(4):1217--1255.

\bibitem[Borch, 1962]{borch1962}
Borch, K. (1962).
\newblock {Equilibrium in a reinsurance market}.
\newblock {\em Econometrica}, 30(3):424--444.

\bibitem[Carlier and Dana, 2003]{carlier2003core}
Carlier, G. and Dana, R.-A. (2003).
\newblock Core of convex distortions of a probability.
\newblock {\em Journal of Economic Theory}, 113(2):199--222.

\bibitem[Castagnoli et~al., 2022]{castagnoli2022star}
Castagnoli, E., Cattelan, G., Maccheroni, F., Tebaldi, C., and Wang, R. (2022).
\newblock Star-shaped risk measures.
\newblock {\em Operations Research}, 70(5):2637--2654.

\bibitem[Chateauneuf et~al., 2000]{chateauneuf2000optimal}
Chateauneuf, A., Dana, R.~A., and Tallon, J.~M. (2000).
\newblock Optimal risk-sharing rules and equilibria with choquet-expected-utility.
\newblock {\em Journal of Mathematical Economics}, 34(2):191--214.

\bibitem[Cui et~al., 2013]{cui2013optimal}
Cui, W., Yang, J., and Wu, L. (2013).
\newblock Optimal reinsurance minimizing the distortion risk measure under general reinsurance premium principles.
\newblock {\em Insurance: Mathematics and Economics}, 53(1):74--85.

\bibitem[Dana, 2004]{dana2004ambiguity}
Dana, R.~A. (2004).
\newblock Ambiguity, uncertainty aversion and equilibrium welfare.
\newblock {\em Economic Theory}, 23:569--587.

\bibitem[De~Castro and Chateauneuf, 2011]{decastro2011ambiguity}
De~Castro, L.~I. and Chateauneuf, A. (2011).
\newblock Ambiguity aversion and trade.
\newblock {\em Economic Theory}, 48:243--273.

\bibitem[Denuit et~al., 2023]{denuitetal2023comonotonicity}
Denuit, M., Dhaene, J., Ghossoub, M., and Robert, C. (2023).
\newblock {Comonotonicity and Pareto optimality with application to collaborative insurance}.
\newblock SSRN: 4337038.

\bibitem[Dhaene et~al., 2002]{dhaene2002concept}
Dhaene, J., Denuit, M., Goovaerts, M.~J., Kaas, R., and Vyncke, D. (2002).
\newblock The concept of comonotonicity in actuarial science and finance: Theory.
\newblock {\em Insurance: Mathematics and Economics}, 31(1):3--33.

\bibitem[Dhaene et~al., 2012]{dhaene2012remarks}
Dhaene, J., Kukush, A., Linders, D., and Tang, Q. (2012).
\newblock Remarks on quantiles and distortion risk measures.
\newblock {\em European Actuarial Journal}, 2:319--328.

\bibitem[Embrechts et~al., 2020]{embrechts2020quantile}
Embrechts, P., Liu, H., Mao, T., and Wang, R. (2020).
\newblock Quantile-based risk sharing with heterogeneous beliefs.
\newblock {\em Mathematical Programming}, 181:319--347.

\bibitem[Embrechts et~al., 2018]{embrechts2018quantile}
Embrechts, P., Liu, H., and Wang, R. (2018).
\newblock Quantile-based risk sharing.
\newblock {\em Operations Research}, 66(4):936--949.

\bibitem[Filipovi{\'c} and Svindland, 2008]{filipovic2008optimal}
Filipovi{\'c}, D. and Svindland, G. (2008).
\newblock Optimal capital and risk allocations for law-and cash-invariant convex functions.
\newblock {\em Finance and Stochastics}, 12:423--439.

\bibitem[F{\"o}llmer and Schied, 2002]{follmer2002convex}
F{\"o}llmer, H. and Schied, A. (2002).
\newblock Convex measures of risk and trading constraints.
\newblock {\em Finance and Stochastics}, 6:429--447.

\bibitem[F{\"o}llmer and Schied, 2011]{follmer2011stochastic}
F{\"o}llmer, H. and Schied, A. (2011).
\newblock {\em Stochastic finance: An introduction in discrete time}.
\newblock Walter de Gruyter.

\bibitem[Ghossoub, 2019]{ghossoub2019optimal}
Ghossoub, M. (2019).
\newblock Optimal insurance under rank-dependent expected utility.
\newblock {\em Insurance: Mathematics and Economics}, 87:51--66.

\bibitem[Ghossoub and Zhu, 2024]{ghossoubzhu2024}
Ghossoub, M. and Zhu, M.~B. (2024).
\newblock Efficiency in pure-exchange economies with risk-averse monetary utilities.
\newblock {\em arXiv preprint arXiv:2406.02712}.

\bibitem[Herings and Zhan, 2022]{HeringsZhany2022}
Herings, P. J.~J. and Zhan, Y. (2022).
\newblock {Competitive equilibria in incomplete markets with risk loving preferences}.
\newblock SSRN: 4245754.

\bibitem[Jouini et~al., 2008]{jouini2008optimal}
Jouini, E., Schachermayer, W., and Touzi, N. (2008).
\newblock Optimal risk sharing for law invariant monetary utility functions.
\newblock {\em Mathematical Finance}, 18(2):269--292.

\bibitem[Landsberger and Meilijson, 1994]{landsberger1994co}
Landsberger, M. and Meilijson, I. (1994).
\newblock Co-monotone allocations, bickel-lehmann dispersion and the arrow-pratt measure of risk aversion.
\newblock {\em Annals of Operations Research}, 52:97--106.

\bibitem[Lauzier et~al., 2023a]{lauzier2023pairwise}
Lauzier, J.~G., Lin, L., and Wang, R. (2023a).
\newblock Pairwise counter-monotonicity.
\newblock {\em Insurance: Mathematics and Economics}, 111:279--287.

\bibitem[Lauzier et~al., 2023b]{lauzier2023risk}
Lauzier, J.~G., Lin, L., and Wang, R. (2023b).
\newblock Risk sharing, measuring variability, and distortion riskmetrics.
\newblock {\em arXiv preprint arXiv:2302.04034}.

\bibitem[Lauzier et~al., 2024]{lauzier2024negatively}
Lauzier, J.~G., Lin, L., and Wang, R. (2024).
\newblock Negatively dependent optimal risk sharing.
\newblock {\em arXiv preprint arXiv:2401.03328}.

\bibitem[Liu et~al., 2020]{liu2020}
Liu, P., Wang, R., and Wei, L. (2020).
\newblock Is the inf-convolution of law-invariant preferences law-invariant?
\newblock {\em Insurance: Mathematics and Economics}, 91:144--154.

\bibitem[Marinacci and Montrucchio, 2004]{marinacci2004introduction}
Marinacci, M. and Montrucchio, L. (2004).
\newblock Introduction to the mathematics of ambiguity.
\newblock In {\em Uncertainty in Economic Theory}, pages 46--107. Routledge.

\bibitem[McNeil et~al., 2015]{mcneil2015quantitative}
McNeil, A.~J., Frey, R., and Embrechts, P. (2015).
\newblock {\em Quantitative risk management: Concepts, techniques and tools-revised edition}.
\newblock Princeton University Press.

\bibitem[Principi et~al., 2023]{principi2023antimonotonicity}
Principi, G., Wakker, P.~P., and Wang, R. (2023).
\newblock Antimonotonicity for preference axioms: The natural counterpart to comonotonicity.
\newblock {\em arXiv preprint arXiv:2307.08542}.

\bibitem[Ravanelli and Svindland, 2014]{RavanelliSvindland2014}
Ravanelli, C. and Svindland, G. (2014).
\newblock {Pareto Optimal Allocations for Law Invariant Robust Utilities on $L^1$}.
\newblock {\em Finance and Stochastics}, 18:249–269.

\bibitem[Rothschild and Stiglitz, 1970]{RothschildStiglitz1970}
Rothschild, M. and Stiglitz, J. (1970).
\newblock {Increasing Risk: I. A Definition}.
\newblock {\em Journal of Economic Theory}, 2(3):225--243.

\bibitem[R{\"u}schendorf, 2013]{ruschendorf2013mathematical}
R{\"u}schendorf, L. (2013).
\newblock {\em Mathematical risk analysis}.
\newblock Springer.

\bibitem[Tsanakas and Christofides, 2006]{tsanakas2006risk}
Tsanakas, A. and Christofides, N. (2006).
\newblock Risk exchange with distorted probabilities.
\newblock {\em ASTIN Bulletin: The Journal of the IAA}, 36(1):219--243.

\bibitem[Tversky and Kahneman, 1992]{tversky1992advances}
Tversky, A. and Kahneman, D. (1992).
\newblock Advances in prospect theory: Cumulative representation of uncertainty.
\newblock {\em Journal of Risk and Uncertainty}, 5:297--323.

\bibitem[Wang et~al., 2020a]{wang2020distortion}
Wang, Q., Wang, R., and Wei, Y. (2020a).
\newblock Distortion riskmetrics on general spaces.
\newblock {\em ASTIN Bulletin: The Journal of the IAA}, 50(3):827--851.

\bibitem[Wang et~al., 2020b]{wang2020characterization}
Wang, R., Wei, Y., and Willmot, G.~E. (2020b).
\newblock Characterization, robustness, and aggregation of signed choquet integrals.
\newblock {\em Mathematics of Operations Research}, 45(3):993--1015.

\bibitem[Wang, 2000]{wang2000class}
Wang, S.~S. (2000).
\newblock A class of distortion operators for pricing financial and insurance risks.
\newblock {\em Journal of Risk and Insurance}, 67:15--36.

\bibitem[Weber, 2018]{weber2018solvency}
Weber, S. (2018).
\newblock Solvency ii, or how to sweep the downside risk under the carpet.
\newblock {\em Insurance: Mathematics and Economics}, 82:191--200.

\bibitem[Wilson, 1968]{wilson1968theory}
Wilson, R. (1968).
\newblock {The theory of syndicates}.
\newblock {\em Econometrica}, 36:119--132.

\bibitem[Wu and Gonzalez, 1996]{wu1996curvature}
Wu, G. and Gonzalez, R. (1996).
\newblock Curvature of the probability weighting function.
\newblock {\em Management Science}, 42(12):1676--1690.

\bibitem[Yaari, 1987]{yaari1987dual}
Yaari, M.~E. (1987).
\newblock The dual theory of choice under risk.
\newblock {\em Econometrica}, 55:95--115.

\end{thebibliography}

\newpage
\appendix
\section{Another example of a portfolio manager's problem}\label{sec:port_exp}
As demonstrated in Example \ref{ex:wang} from Section \ref{sec:port}, agents' risk preferences are modeled using Wang transform. The model indicates that the optimal value of $\lambda^*$ for a group with more agents tends to be smaller than that for a group with fewer agents. However, this is not always the case; the opposite outcome can also occur, depending on the specific distortion function applied. 
In this section, we construct a new distortion function to model the agents' risk preferences, showing that, in contrast to the previous model, the optimal value of $\lambda^*$ for a group with more agents tends to be larger than that for a group with fewer agents.

The dual distortion function $\tilde{h}$ is constructed as:
\begin{align*}
    \tilde{h}(t)=
    \frac{k^{\alpha}}{1+(k-1)\alpha}t^\alpha \mathbb{1}_{\left\{t\leq \frac{1}{k}\right\}}+ \left(\frac{k \alpha}{1+(k-1)\alpha}t + \frac{1- \alpha}{1+(k-1)\alpha}\right)\mathbb{1}_{\left\{t> \frac{1}{k}\right\}},
\end{align*}
where $\alpha\in \mathbb{R}_{+}$ and $k$ is a positive integer. The distortion function is thus defined as $h(t)=1-\tilde{h}(1-t)$.
It is trivial to see $h$ is convex with $\alpha \in (0,1]$ and $h$ is concave with $\alpha \in [1,\infty)$.

Here we choose $k=10$. Numerical results are presented below.
\begin{figure}[ht]
    \centering
    \includegraphics[width=8.5cm]{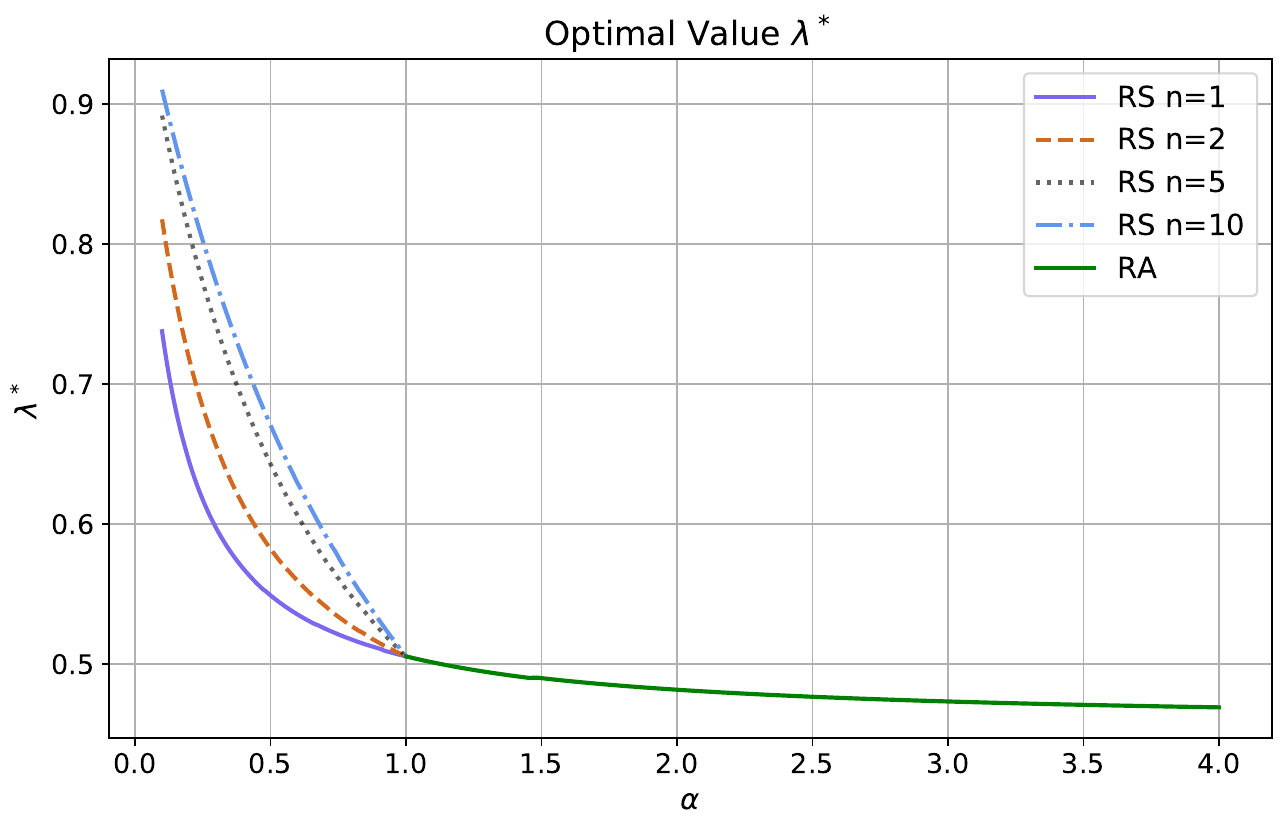}
    \caption{Optimal value of $\lambda^*$ with $1-\tilde{h}(1-t)$.}
    \label{fig:new}
    \end{figure}

\section{A counter-example}
\label{sec:counter-example}

Below we present a counter-example showing $\dboxminus_{i=1}^{3}\rho_{i}\ne  \rho_1 \dboxminus \rho_2 \dboxminus \rho_3  $.
\begin{example}\label{exa:repeat}
    Define a probability space $(\Omega^{\prime}, \mathcal{F}^{\prime}, \mathbb{P}^{\prime})$ where $\Omega^{\prime}=\left\{\omega_1, \omega_2,\omega_3,\omega_4\right\}$, 
   $ \mathcal{F}^{\prime}$ is the power set of $\Omega$, and $\mathbb{P}^{\prime}$ is such that $\mathbb{P}^{\prime}(\omega_i)=1/4$ for $i=1, \dots, 4$. Define three distributions $F_{1}=\text{Bernoulli(1/2)}$, and $F_{2}= F_{3}=1/2 \times \text{Bernoulli(1/4)}$. 
    Suppose that three agents have the risk measures given by 
    \begin{align*}
        \rho_{i}(X)=1-\mathbb{1}_{\left\{X\sim F_{i}\right\}}, \  i= 1,2,3 \  \text{and} \  X\in \mathcal{X}.
    \end{align*}
    Then we define four random variables $X=\mathbb{1}_{\left\{\omega_1, \omega_2,\omega_3\right\}}$, 
    $X_{1}=\mathbb{1}_{\left\{\omega_1,\omega_2\right\}}$, and $X_{2}=X_{3}=1/2 \times  \mathbb{1}_{\left\{\omega_3\right\}}$. Clearly, $(X_1, X_2, X_3)$ is an allocation of $X$. It is straightforward to show both $(X_1, X_2)$ and $(X_1+X_2, X_3)$ are counter-monotonic, whereas $(X_2, X_3)$ is comonotonic. It follows that
    \begin{align*}
        \rho_1 \dboxminus \rho_2 \dboxminus \rho_3 (X)\leq \rho_1(X_1)+\rho_2(X_2)+\rho_3(X_3)=0.
    \end{align*}
    Hence, we have $\rho_1 \dboxminus \rho_2 \dboxminus \rho_3 (X)=0$ since $\rho_i$ are non-negative for all $i\in [3]$. 
    Furthermore, 
    for any $(Y_1, Y_2, Y_3)\in \mathbb{A}_{3}^{-}(X)$, 
    at least one of $\rho_i (Y_i)$ is equal to 1 as there does not exist a counter-monotonic allocation such that $Y_i \sim F_i$, $i\in [3]$. Hence,  $\dboxminus_{i=1}^{3}\rho_{i}(X)\neq 0 > \rho_1 \dboxminus \rho_2 \dboxminus \rho_3 (X)$. 
    \end{example}

\end{document}